\documentclass[sigconf,10pt]{acmart} 
\usepackage{algorithm}
\usepackage{appendix}
\usepackage{clrscode}
\usepackage{amsmath}
\usepackage{amssymb}
\usepackage{appendix}
\usepackage{caption}
\usepackage{balance}
\usepackage{color}
\usepackage{comment}
\usepackage{graphics}
\usepackage{graphicx}
\usepackage{listings}
\usepackage{mathtools}
\usepackage{multirow}
\usepackage{shortvrb}
\usepackage{subcaption}
\usepackage{tabularx}
\usepackage{url}
\usepackage{mdwlist}
\usepackage{tikz}
\usepackage{chngcntr}
\usepackage{booktabs}

\newcommand{\system}{{\sc FITing-Tree}}

\DeclareMathOperator*{\argmin}{\arg \min}

\begin{document}
\title{FITing-Tree: A Data-aware Index Structure}


\author{Alex Galakatos$^{1*}$\thanks{*Authors contributed equally}  \ \ \  Michael Markovitch$^{1*}$ \ \ \ Carsten Binnig$^2$ \\
       Rodrigo Fonseca$^1$  \ \ \ Tim Kraska$^3$
       \\
       \small
       $^1$Brown University
       \{first\_last\}@brown.edu \ \ 
       $^2$TU Darmstadt \{first.last@cs.tu-darmstadt.de\} \ \ 
       $^3$MIT CSAIL \{last@mit.edu\}
       }






\renewcommand{\shortauthors}{A. Galakatos et al.}
\renewcommand{\authors}{Alex Galakatos, Michael Markovitch, Carsten Binnig, Rodrigo Fonseca, Tim Kraska}

\copyrightyear{2019} 
\acmYear{2019} 
\setcopyright{acmcopyright}
\acmConference[SIGMOD '19]{2019 International Conference on Management of Data}{June 30-July 5, 2019}{Amsterdam, Netherlands}
\acmBooktitle{2019 International Conference on Management of Data (SIGMOD '19), June 30-July 5, 2019, Amsterdam, Netherlands}
\acmPrice{15.00}
\acmDOI{10.1145/3299869.3319860}
\acmISBN{978-1-4503-5643-5/19/06}

\fancyhead{}

\begin{abstract}
Index structures are one of the most important tools that DBAs leverage to improve the performance of analytics and transactional workloads.
However, building several indexes over large datasets can often become prohibitive and consume valuable system resources.
In fact, a recent study showed that indexes created as part of the TPC-C benchmark can account for 55\% of the total memory available in a modern DBMS.
This overhead consumes valuable and expensive main memory, and limits the amount of space available to store new data or process existing data.

In this paper, we present \system{}, a novel form of a learned index which uses piece-wise linear functions with a bounded error specified at construction time.
This error knob provides a tunable parameter that allows a DBA to FIT an index to a dataset and workload by being able to balance lookup performance and space consumption.
To navigate this tradeoff, we provide a cost model that helps determine an appropriate error parameter given either (1) a lookup latency requirement (e.g., $500ns$) or (2) a storage budget (e.g., $100MB$).
Using a variety of real-world datasets, we show that our index is able to provide performance that is comparable to full index structures while reducing the storage footprint by orders of magnitude.
\end{abstract}

\maketitle

\section{Introduction}


Tree-based index structures (e.g., B+ trees) are one of the most important tools that DBAs leverage to improve the performance of analytics and transactional workloads.
However, for main-memory databases, tree-based indexes can often consume a significant amount of memory.
In fact, a recent study~\cite{reducing_storage} shows that the indexes created for typical OLTP workloads can consume up to 55\% of the total memory available in a state-of-the-art in-memory DBMS.
This overhead not only limits the amount of space available to store new data but also reduces space for intermediates that can be helpful when processing existing data.

To reduce the storage overhead of B+ trees, various compression schemes have been developed~\cite{bftree,graefe,prefixtrees,superscalar}.
The main idea behind these techniques is to remove the redundancy that exists among keys and/or to reduce the size of each key inside a node of the index.
For example, prefix and suffix truncation can be used to store common parts of keys only once per index node, reducing the total size of the tree.
Additionally, more expensive compression techniques like Huffmann coding can be applied within each node but come at a higher runtime cost since pages must be decompressed to search for an item. 

Although each of the previously mentioned compression schemes reduce the size of an index node, the memory footprint of these indexes still grows linearly with the number of distinct keys to be indexed, resulting in indexes that can consume a significant amount of memory.
This observation is especially true for data such as timestamps or sensor readings that are generated in a wide variety of applications (e.g., autonomous vehicles, IoT devices). 
Even worse, the number of unique keys to be indexed for such data types typically grow over time, resulting in indexes that are constantly growing.
Consequently, a DBA has no way to restrict memory consumption other than dropping an index completely.

To tackle this issue, we present \system{}, a novel index structure that compactly captures trends in data using piece-wise linear functions.
Unlike typical indexes which use fixed-size pages on the leaf level that point to the data, \system{} uses piece-wise linear functions to quickly \emph{approximate} the position of an element.
By leveraging the trends within the data, \system{} can reduce the memory consumption of an index by orders of magnitude compared to a traditional B+ tree.
At the core of our index structure is a parameter that specifies the amount of acceptable error (i.e., a constant that is the maximum distance between the predicted and actual position of any key).
Unlike existing index structures, our error parameter allows a DBA to FIT an index to a given scenario and balance the lookup performance and space consumption of an index.
To navigate this tradeoff, we also present a cost model that helps a DBA choose an appropriate error term given either (1) a lookup latency requirement (e.g., $500$ns) or (2) a storage budget (e.g., $100$MB).

In the basic version of \system{}, we assume that the table data to be indexed is sorted by the index key (i.e., clustered index) but we also show how our techniques extend to secondary (i.e., non-clustered) indexes.
Using a variety of real-world datasets, we show that our index structure provides performance that is comparable to full and fixed-page indexes (even for a worst-case dataset) while reducing the storage footprint by orders of magnitude.

Using linear functions to approximate the distribution makes \system{} a form of a learned index~\cite{learnedindexes}. 
However, in contrast to the initially proposed techniques, our approach allow us to (1) bound for the worst-case lookup performance, (2) efficiently support insert operations, and (3) enable paging (i.e., the entire data does not have to reside in a contiguous memory region).
Furthermore, although the problem of approximating distributions using piece-wise functions is also not new~\cite{ESCH196985,Braess1971,smartgrid,shatkay1996approximate,keogh2001online,FU2011164,simsearch,Xu:2012:AAO:2247596.2247620,Liu:2008:NOM:1477069.1477485,neumann}, none of these techniques have been applied to indexing and therefore do not consider operations that indexes must support.

Another interesting observation is that our compression scheme in \system{} is orthogonal to node-level compression techniques such as the previously mentioned prefix/suffix truncation.
In other words, since \system{} internally uses a tree structure for inner nodes, we can still apply these techniques to further reduce an index's size.


In summary, we make the following contributions:
\begin{itemize*}
\item We propose \system{}, a novel index structure that leverages properties about the underlying data distribution to reduce the size of an index.
\item We present and analyze an efficient segmentation algorithm that incorporates a tunable error parameter that allows DBAs to balance the lookup performance and space footprint of our index.
\item We propose a cost model that helps a DBA determine an appropriate error threshold given either a latency or storage requirement.
\item Using several real-world datasets, we show that our index provides similar (or in some cases even better) performance compared to existing index structures while consuming orders of magnitude less space.
\end{itemize*}

The remainder of the paper is organized as follows.
In Section \ref{sec:overview}, we first present an overview of our new index structure called \system{}.
Afterwards, we discuss the main index operations: bulk loading (Section \ref{sec:bulk_loading}), lookups (Section \ref{sec:lookups}) and inserts (Section \ref{sec:inserts}).
Section \ref{sec:cost_model} then introduces our cost model that allows a DBA to balance the lookup performance and the space consumption of a \system{}.
Finally, in Section \ref{sec:experiments} we discuss the results of our evaluation on real and synthetic datasets, summarize related work in Section \ref{sec:related}, and finally conclude in Section \ref{sec:conclusion}.

\section{Overview}
\label{sec:overview}

At a high level, indexes (and B+ trees over sorted attributes in particular) can be represented by a function that maps a key (e.g., a timestamp) to a storage location.
Using this representation, \system{} partitions the key space into a series of disjoint linear segments that approximate the true function, since it is (generally) not possible to fully model the underlying data distribution.
At the core of this process is a tunable error threshold which represents the maximum distance that the predicted location of any key inside a segment is from its actual location.
Instead of storing all values in the key space, \system{} stores only (1) the starting key of each linear segment and (2) the slope of the linear function in order to compute a key's approximate position using linear interpolation.

In the following, we first discuss how we can use functions to map key values to storage locations.
Then, we discuss how we leverage this function representation to efficiently implement our index structure on top of a B+ tree for clustered indexes.
Finally, we show how our ideas can also be applied to compress secondary indexes.


\begin{figure}
\begin{center}
\vspace{-2mm}
\includegraphics[width=.9\columnwidth]{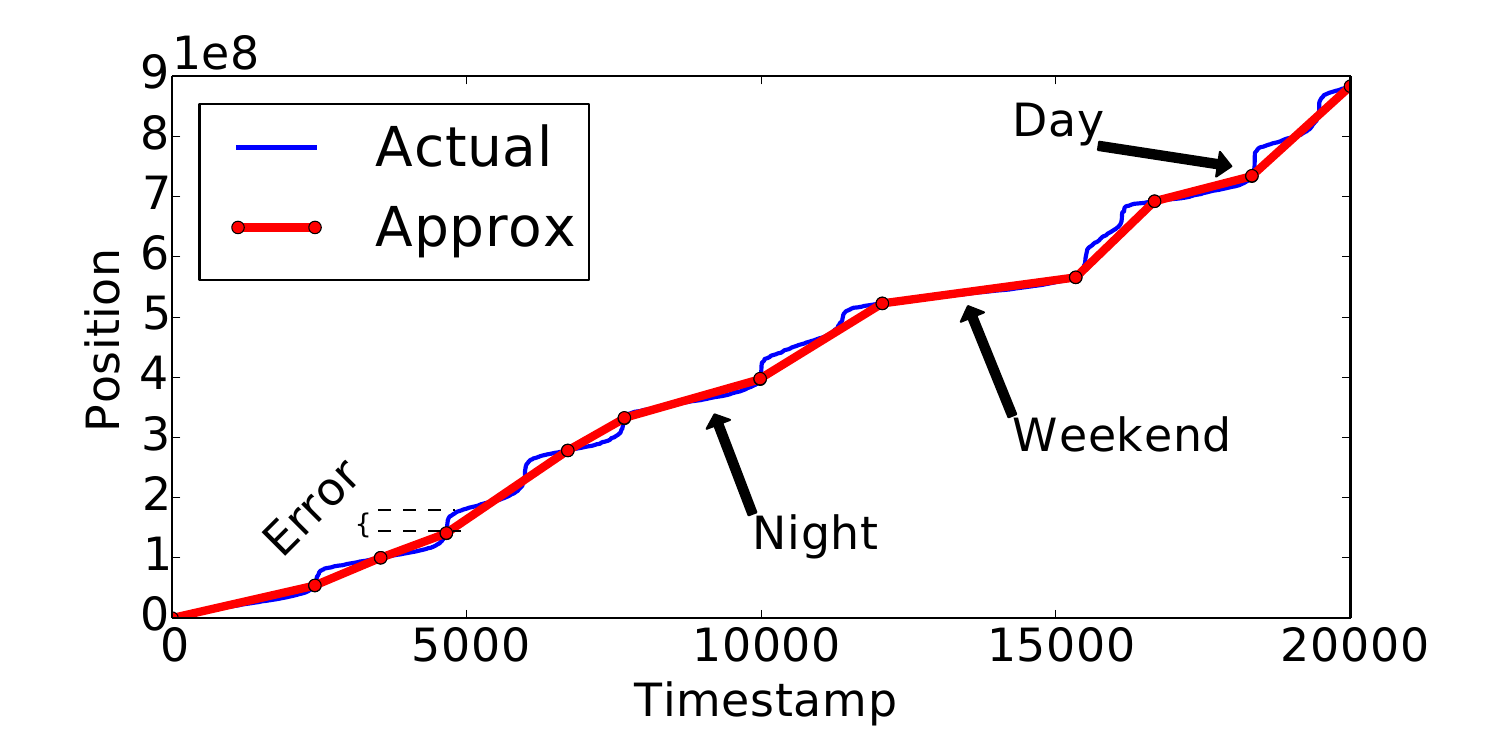}
  \vspace{-2mm}
\caption{Key to position mapping for IoT data.}
  \vspace{-3mm}
\label{fig:timestamp_mapping}
\end{center}
\end{figure}

\subsection{Function Representation}
\label{sec:overview:function_representation}
One key insight to our approach is that we can abstractly model an index as a monotonically increasing function that maps keys (i.e., values of the indexed attribute) to storage locations (i.e., its page and the offset within that page).
To explain this intuition, assume that all keys to be indexed are stored in a sorted array, allowing us to use an element's position in the array as its storage location.



As an example, consider the IoT dataset~\cite{mgbench}, which contains events from various devices (e.g., door sensors, motion sensors, power monitors) installed throughout a university building.
In this dataset, the data is sorted by the timestamp of an event, allowing us to construct a function that maps each timestamp (i.e., key) to its position in the dataset (i.e., position in a sorted array), as shown in Figure~\ref{fig:timestamp_mapping}.
Unsurprisingly, since the installed IoT devices monitor human activity, the timestamps of the recorded actions follow a pattern (e.g., there is little activity during the weekend and at night).

Since a function that represents an index can be arbitrarily complex and data-dependent, the precise function that maps keys to positions may not be possible to learn and is expensive to build and update.
Therefore, our goal is to approximate the function that represents the mapping of a key to a position.

To compactly capture trends that exist in the data while being able to efficiently build a new index and handle updates, we use a series of piece-wise linear functions to approximate an arbitrary function.
As shown in Figure~\ref{fig:timestamp_mapping}, for example, our segmentation algorithm (described further in Section~\ref{sec:bulk_loading}) partitions the timestamp values into several linear segments that are able to accurately reflect the various trends that exist in the data (e.g., less activity during the weekend).
Since the approximation captures trends in the data, it is agnostic to key density (a trend with sparse keys can be captured as well as a trend with dense keys).

Although more complex functions (e.g., higher order polynomials) can be used to approximate the true function, piece-wise linear approximation is significantly less expensive to compute.
This dramatically reduces (1) the initial index construction cost, and (2) improves insert latency for new items (see Section~\ref{sec:inserts}).

The resulting piece-wise linear approximation, however, is not precise (i.e., a key's predicted location is not necessarily its true position).
We therefore define the \textit{error} associated with our approximation as the maximum distance between the actual and predicted location of any key, as shown below, where $pred\_pos(k)$ and $true\_pos(k)$ return the predicated and actual position of an element $k$ respectively.

\vspace{-2ex}
\begin{equation}
    error = \max(| pred\_pos(k) - true\_pos(k)|) \ \forall \ k \ \in \ keys
\end{equation}

This formulation allows us to define the core building block of \system{}, a \textit{segment}.
A segment is a contiguous region of a sorted array for which any key is no more than a specified error threshold from its interpolated position.
Depending on the data distribution and the error threshold, the segmentation process will yield a different number of segments that approximate the underlying data.
Therefore, importantly, the error threshold enables us to balance memory consumption and performance.
After the segmentation process, \system{} stores the boundaries and slope of each segment (instead of each individual key) in a B+ tree, reducing the overall memory footprint of the index.


\subsection{\system{} Design}

As previously mentioned, our segmentation process partitions the key space of an attribute into disjoint linear segments such that the predicted position of any key inside a segment is no more than a bounded distance away from the key's true position.
\system{} organizes these segments in a tree to efficiently support insert and lookup operations.

In the following, we first discuss clustered indexes, where records are already sorted by the key that is being indexed.
Afterwards, we show how our technique can be extended to provide similar benefits for secondary indexes.

\subsubsection{Clustered Indexes}

\begin{figure}
\begin{center}
\includegraphics[width=0.78\columnwidth]{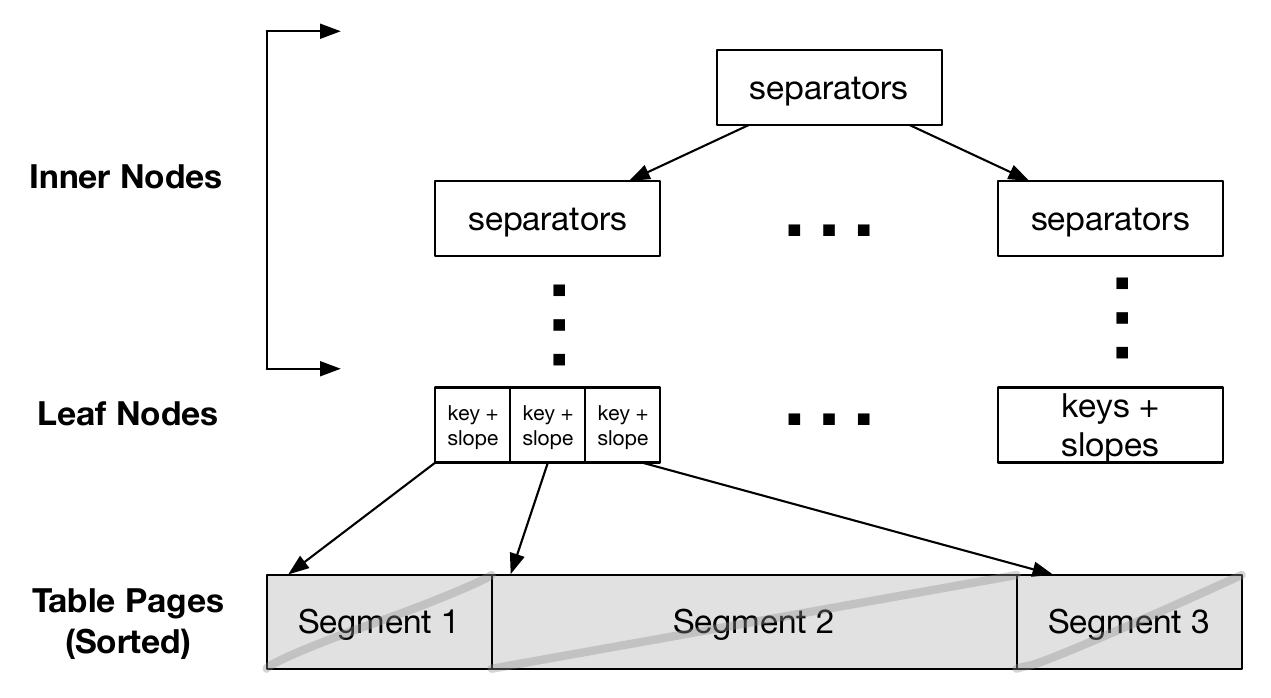}
\caption{A clustered \system{} index.}
  \vspace{-6mm}
\label{fig:clustered_tree}
\end{center}
\end{figure}

In a traditional clustered B+ tree, the table data is stored in fixed-size pages and the leaf level of the index contains only the first key of each of these table pages.
Unlike a clustered B+ tree, in a \system{}, the table data is partitioned into variable-sized segments (pages) that satisfy the given error threshold.
Each segment is essentially a fixed-size array, but successive segments can be allocated independently (i.e., non contiguously).

Figure~\ref{fig:clustered_tree} shows the structure of a clustered \system{} index.
As shown, the underlying data is partitioned into a series of variable-sized segments that approximate the distribution of keys to be indexed.
Depending on the error parameter and the data distribution, several consecutive keys can thus be summarized into a single segment.
Details of the segmentation algorithm that divides the table data into variable-sized segments are discussed in Section~\ref{sec:bulk_loading}.

Unlike a traditional B+ tree, each leaf node in a \system{} stores the segment's slope, starting key, and a pointer to a segment.
This allows us to use interpolation search in each segment since the data within this segment is approximated by a linear function given by the slope.

The inner nodes of a \system{} are the same as a B+ tree (i.e., lookup and insert operations are identical to a normal B+ tree).
However, once a lookup or an insert reaches the leaf level, \system{} needs to perform additional work.
For lookups, we need to use the slope and the distance to the starting key to calculate the key's approximate position (offset in the segment).
Since the resulting position is approximate, \system{} must then perform a local search (e.g., binary, linear) to find the item, discussed further in Section~\ref{sec:lookups}.

Insert operations also require additional work upon reaching a leaf level page, since we must ensure that the error threshold is always satisfied.
Therefore, we present two different insertion strategies (described in detail in Section~\ref{sec:inserts}).
The first strategy performs in-place updates to the segment (as a baseline) while the second strategy uses a more advanced buffer-based strategy to hold inserted data items.

\begin{figure}
\begin{center}
\includegraphics[width=0.78\columnwidth]{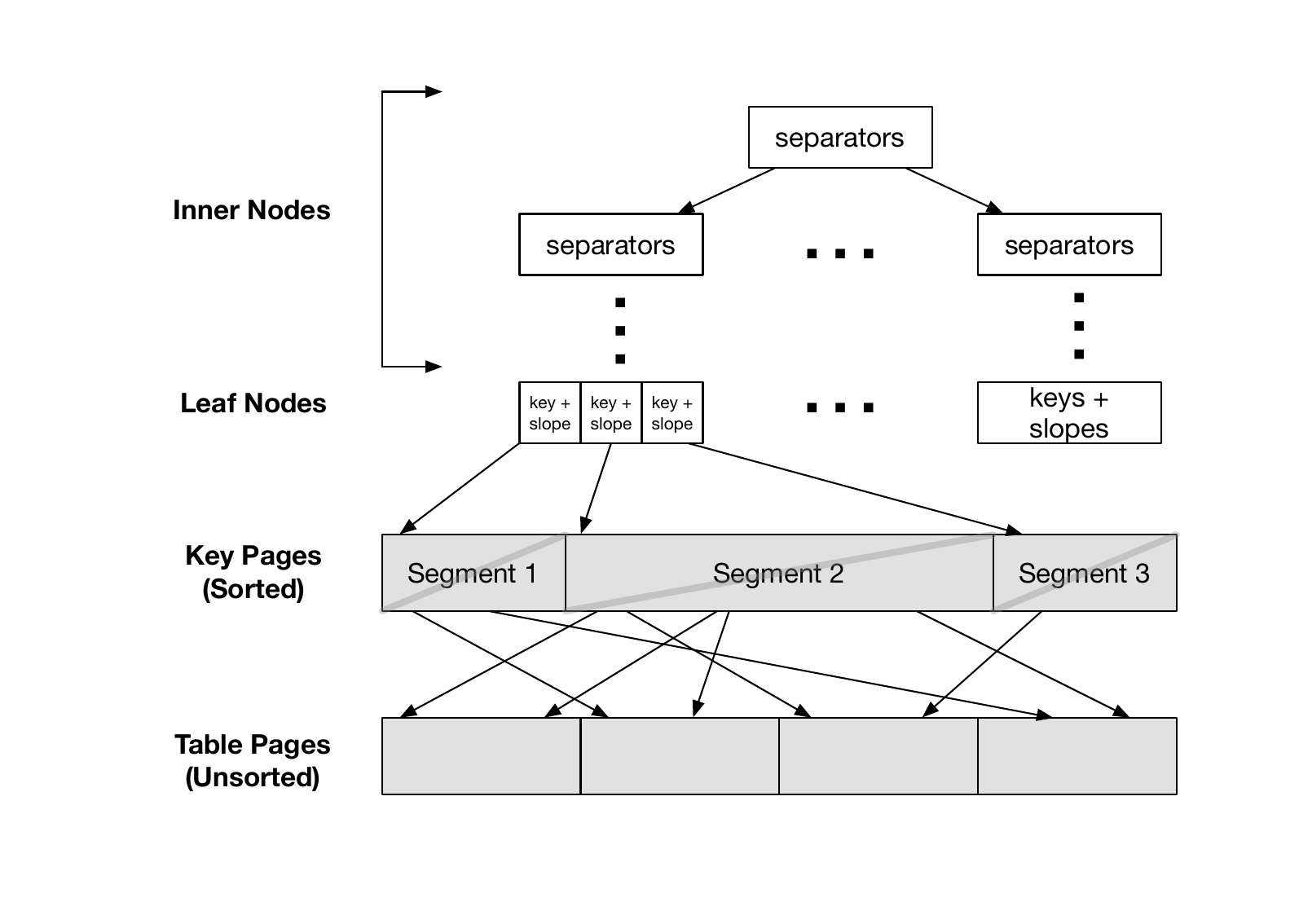}
\caption{A non-clustered \system{} index.}
\label{fig:non_clustered_tree}
\end{center}
\end{figure}


Finally, instead of internally using a B+ tree to locate the segment for a key, \system{} could instead use any other index structure.
For example, if the workload is read-only, other high performance index structures (e.g., FAST~\cite{FAST}) can be used.
In Section~\ref{sec:experiments:other_indexes} we show how \system{} performs when using different internal data structures, including FAST.



\subsubsection{Non-clustered Indexes}
\label{sec:non-clustered}

Secondary indexes can dramatically improve the performance of queries involving selections over a non-primary key attribute.
Without secondary indexes, these queries must examine all tuples, which is often prohibitive.
Unlike a clustered index, a non-primary key attribute is not sorted and may contain duplicates.

The primary difference between a clustered and an non-clustered \system{} is that a non-clustered \system{} requires an additional ``indirection layer'' (called ``Key Pages'' in Figure~\ref{fig:non_clustered_tree}).
This layer is essentially an array of pointers that is the same size as the data but is sorted by the key that is being indexed.
For example, the first position in this indirection layer will contain a pointer to the smallest value of the key being indexed.
Note that a secondary B+ tree that uses fixed-size paging also requires this indirection layer.

The first step in creating a non-clustered \system{} is to build the indirection layer by sorting the data by the indexed key (e.g., temperature, age) and materializing the array of pointers to each data item in the sorted order.
Next, like in the clustered index case, the segmentation algorithm scans the indirection layer and produces a valid set of segments that are then inserted into the upper level tree.

All operations on a non-clustered \system{} internally operate on the indirection layer.
For example, for a lookup, the returned position of the data item is its position in the indirection layer.
Then, to access the value, \system{} follows the pointer in the indirection layer at the predicted position.

Although the sorted level of key pages in a non-clustered \system{} introduces additional overhead compared to a clustered \system{} index, this overhead occurs in any non-clustered (secondary) index.
However, as we show in our experiments, a non-clustered \system{} is significantly smaller than a non-clustered B+ tree with fixed-size pages since it has fewer leaf and internal nodes.




\section{Segmentation}
\label{sec:bulk_loading}

In the following, we describe how \system{} partitions the key space of an attribute into variable-sized segments that satisfy the specified error.
After this process, each segment is inserted into a B+ tree to enable efficient insert and lookup operations, described further in Section~\ref{sec:lookups} and Section~\ref{sec:inserts}.

\subsection{Design Choices}



\begin{figure}
\centering
\resizebox{4cm}{3cm}{
\begin{tikzpicture}[domain=0:4] 
\draw[->] (0,0) -- (4.2,0) node[below] {$key$}; 
    \draw[->] (0,0) -- (0,4.2) node[left] {$loc$};
    \filldraw (0.6,0.6) circle (2pt) node[align=center,   below] {$\left(x_1,y_1\right)$};
    \filldraw (1.2,3.4) circle (2pt) node[align=center,   above] {$\left(x_2,y_2\right)$};
    \filldraw (3.5,3.5) circle (2pt) node[align=center,   above] {$\left(x_3,y_3\right)$};
    \draw [dashed] (0.6,0.6) -- (3.5,3.5);
    \draw [<->, red, very thick] (1.2,1.2) -- (1.2,3.35);
    \draw (1.25,2.5) node[right] {$>error$};
\end{tikzpicture}
}
\caption{A segment from $\left(x_1,y_1\right)$ to $\left(x_3,y_3\right)$ is not valid if $\left(x_2,y_2\right)$ is further than $error$ from the interpolated line.}
\label{fig:max_segment}
\end{figure}

A common objective when fitting a function is to minimize the least square error (minimizing the second error norm $E_2$).
Unfortunately, such an objective does not provide a guarantee for the maximal error and therefore does not provide a bound on the number of locations which must be scanned after interpolating a key's position. 
Therefore, our objective is to satisfy a maximal error ($E_\infty$), demonstrated in Figure~\ref{fig:max_segment}.

While several optimal (in the number of produced segments) piece-wise linear approximation algorithms exist, these techniques are prohibitively expensive (e.g., a dynamic programming algorithm~\cite{piecewise} has a runtime of $O(n^3)$ using $O(n)$ memory).
Second, most existing online piece-wise linear approximation algorithms~\cite{Elmeleegy:2009:OPL:1687627.1687645,keogh2001online} have a high storage complexity and/or do not guarantee a maximal error.
Lastly, even linear time algorithms may not be efficient enough since multiplicative constants have a significant effect.

Therefore, to be able to efficiently (1) construct the index, and (2) support inserts, we need a highly efficient one-pass linear algorithm.
This focus on efficiency led us to choose linear piece-wise functions, since higher order approximations often incur additional costs as previously discussed.

\begin{algorithm}[t]
\scriptsize
\caption{\textsc{ShrinkingCone} Segmentation}
\label{alg:streaming_alg}
\begin{codebox}
    \li $sl_{high} \gets \infty$
    \li $sl_{low} \gets 0$
    \li the first key is the segment origin
    \li \For every $k\in$ keys (in increasing order)
        \li \Do
        \If $k$ is inside the cone:
            \li \Then update $sl_{high}$
            \li update $sl_{low}$
            \li \Else key $k$ is the origin of a new segment
            \li $sl_{high} \gets \infty$
            \li $sl_{low} \gets 0$
        \End
    \End
\end{codebox}
\end{algorithm}

In the following, we describe a proposed segmentation algorithm, similar to FSW~\cite{Xu:2012:AAO:2247596.2247620, Liu:2008:NOM:1477069.1477485}, which is linear in runtime, has low constant memory usage, and guarantees a maximal error in each segment.
Importantly, though, we address (1) how to extend these techniques to indexing, including looking up and inserting data items, (2) prove that, in the worst case, segments are bounded in size, and (3) analyze the algorithm and compare it to an optimal algorithm.
\subsection{Segment Definition}
\label{sec:bulk_loading:segment_defi}

As previously described, a segment is a region of the key space that can be represented by a linear function whereby all keys are within a bounded distance from their linearly interpolated position.
More specifically, a segment is represented by the first point (first key) and by the last point (the last key) in the segment. 
Using this definition, we can fit a linear function to the locations of keys in the segment (using the start key, end key, and the number of positions). 

Recall that every segment must satisfy the maximal error (i.e., a key's predicated position is at most $error$ number of elements away from its true position).
This leads to an important property (proof in Appendix~\ref{prf:segment-size}) of a \textit{maximal segment} (a segment is maximal when the addition of a key will violate the specified error):
\begin{theorem}
\label{thm:seg_size}
The minimal number of locations covered by a maximal linear segment is $error + 1$.
\end{theorem}

This allows us to quantify how bad a "worst case" (i.e., dataset and error threshold that produce maximal number of segments) can be.
Since the minimum number of locations covered by a maximal segment is bound by the error, the total size of a \system{} is also bounded.
Therefore, in the worst case (every maximal segment covers $error + 1$ locations), a \system{} will not be larger than an index that uses fixed-size pages of size $error$ (e.g., B+ tree).

\begin{figure}
\centering
\resizebox{4cm}{3cm}{
\begin{tikzpicture}[domain=0:4] 
\draw[->] (0,0) -- (4.2,0) node[below] {$key$}; 
    \draw[->] (0,0) -- (4.2,0) node[below] {$key$}; 
    \draw[->] (0,0) -- (0,4.2) node[left] {$loc$};
    \filldraw (0.6,0.6) circle (2pt) node[left] {1};
    \filldraw (3.5,3.5) circle (2pt) node[right] {4};
    \draw [<->, thick] (1.2,0.8) -- (1.2,1.8);
    \draw [blue, dashed] (0.6,0.6) -- (4.2,1.8);
    \draw [blue, dashed] (0.6,0.6) -- (2.4,4.2);
    \draw [<->, thick] (2.4,1) -- (2.4,2);
    \draw [red, dotted, thick] (0.6,0.6) -- (4.2,3.4);
    \draw [red, dotted, thick] (0.6,0.6) -- (4.2,1.8);
    \filldraw [blue] (1.2,1.3) circle (2pt) node[right] {2};
    \filldraw [red] (2.4,1.5) circle (2pt) node[right] {3};
\end{tikzpicture}
}
\caption{ShrinkingCone - Point 1 is the origin of the cone. Point 2 is then added, resulting in the dashed cone. Point 3 is added next, yielding in the dotted cone. Point 4 is outside the dotted cone and therefore starts a new segment.}
\label{fig:streamin_alg}
\vspace{-3mm}
\end{figure}

\subsection{Segmentation Algorithm}
\label{sec:bulk_loading:online}

As previously mentioned, we needed a fast and efficient algorithm rather than an optimal one.
We therefore chose to use a greedy streaming algorithm \textsc{ShrinkingCone} (Algorithm~\ref{alg:streaming_alg}) which, given a starting point (key) of a segment, attempts to maximize the length of a segment while satisfying a given error threshold. \textsc{ShrinkingCone} is similar to FSW \cite{ Liu:2008:NOM:1477069.1477485} but considers only monotonically increasing functions and can produce disjoint segments.
The main idea behind \textsc{ShrinkingCone} is that a new key can be added to a segment if and only if it does not violate the error constraint of any previous key in the segment. 


More specifically, we define a cone by the triple: origin point (the key and its location), high slope ($sl_{high}$), and low slope ($sl_{low}$).
The combination of the starting point and the low slope gives the lower bound of the cone, and the combination of the starting point and the high slope gives the upper bound of the cone.
Intuitively, the cone represents the family of feasible linear functions for a segment starting at the origin of the cone (the high and low slopes represent the range of valid slopes).
When a new key is added to a segment, the high and low slopes are calculated using the key and the key's position plus $error$ and minus $error$ (respectively).
In the update step (lines 6-7 of Algorithm~\ref{alg:streaming_alg}), the lowest high slope and the highest low slope values are then selected (between the newly calculated and previous slopes).
Therefore, the cone either narrows (the high slope decreases and/or the low slope increases), or stays the same. 
If a new key to be added to the segment is outside of the cone, there must exist at least one previous key in the segment for which the error constraint will be violated.
Therefore, a new key that is not inside the cone cannot be included in the segment, and becomes the origin point of the new segment.


Figure~\ref{fig:streamin_alg} illustrates how the cone is updated: point 1 is the origin of the cone. Point 2 updates both the high and low slopes. Point 3 is inside the cone, however it only updates the upper bound of the cone (point 3 is less than $error$ above the lower bound). Point 4 is outside of the updated cone, and therefore will be the first point of a new segment.

\begin{table}
\center
\tiny
\begin{tabular}{lcccr}
     \toprule
     Dataset & error & \textsc{ShrinkingCone} & Optimal & Ratio  \\
     \midrule
     Taxi drop lat & 10 & 5358 & 4996 & 1.07 \\
     Taxi drop lat & 100 & 351 & 271 & 1.29 \\
     Taxi drop lat & 1000 & 51 & 48 & 1.06 \\
     Taxi drop lon & 10 & 1198 & 1138 & 1.05 \\
     Taxi drop lon & 100 & 371 & 325 & 1.14 \\
     Taxi drop lon & 1000 & 40 & 37 & 1.08 \\
     Taxi pick time & 10 & 6238 & 4359 & 1.43 \\
     Taxi pick time & 100 & 165 & 137 & 1.2 \\
     OSM lon & 10 & 7727 & 6027 & 1.28 \\
     OSM lon & 100 & 101 & 63 & 1.6 \\
     Weblogs & 10 & 16961 & 14179 & 1.2 \\
     Weblogs & 100 & 909 & 642 & 1.42 \\
     IoT & 10 & 8605 & 6945 & 1.24 \\
     IoT & 100 &723 & 572 & 1.26 \\
     \bottomrule
\end{tabular}
\vspace{2mm}
\caption{\textsc{ShrinkingCone} compared to optimal.}
\vspace{-10mm}
\label{tab:opt_comparisson}
\end{table}

\subsection{Algorithm Analysis}
\label{sec:alg-analysis}
While the \textsc{ShrinkingCone} algorithm has a runtime of $O(n)$ and only uses a small constant amount memory (to keep track of the cone), it is not optimal. 
Moreover, for a given maximal error and an adversarial dataset the number of segments that it produces can be arbitrarily worse than an optimal algorithm, as we prove in Appendix~\ref{prf:ca}.

Although \textsc{ShrinkingCone} can be arbitrarily worse compared to optimal segmentation for a given maximal error, there is a limit for how bad it can be in practice since we do have a guarantee that a maximal segment covers at least $error + 1$ locations.

The maximum number of segments \textsc{ShrinkingCone} produces is at most $\min\left(\frac{|keys|}{2},\frac{|D|}{error+1}\right)$, where $|D|$ is the size of the dataset.
This guarantee stems from Theorem~\ref{thm:seg_size}: no input with less than 3 keys spanning at least $error + 2$ positions will cause \textsc{ShrinkingCone} to create a new segment.
Thus, compared to traditional B+ trees, in the worst case, \system{} will produce no more segments (pages) than a B+ tree that uses fixed-size pages (of size $error$).


To evaluate \textsc{ShrinkingCone}, we implemented the optimal algorithm (runtime of $O(n^2)$ and memory consumption of $O(n^2)$) using $10^6$ elements from real-world datasets: NYC Taxi Dataset~\cite{NYCtaxi}, OpenStreetMap~\cite{OSMdata}, Weblogs~\cite{mgbench}, and IoT~\cite{mgbench}.
Table~\ref{tab:opt_comparisson} shows the number of segments generated by the optimal algorithm and by \textsc{ShrinkingCone}.
As shown, the number of segments that our algorithm produces is comparable to the number of segments in the optimal case.


\section{Index Lookups}
\label{sec:lookups}

One of the most important operations of an index is to lookup a single key or a range of keys.
However, since each entry in the leaf level of \system{} points to a segment, performing a lookup requires first locating the segment that a key belongs to and then performing a local search inside the segment.
In the following, we first describe how \system{} performs lookup operations for a single key and then show how we can extend this technique to range predicates.

\subsection{Point Queries}
The process of searching a \system{} for a single element involves two steps: (1) searching the tree to find the segment that the element belongs to, and (2) finding the element within a segment.
These steps are outlined in Algorithm~\ref{algo:lookup}.

\subsubsection{Tree Search}
Since, as previously described, each segment is stored in a B+ tree (with its first key as the key and the segment's slope and a pointer to the table page as its value), we must first search the tree to find the segment that a given key belongs to.
To do this, we begin traversing the B+ tree from the root to the leaf, using standard tree traversal algorithms.
These steps, outlined in the \textsc{SearchTree} function of Algorithm~\ref{algo:lookup}, terminate when reaching a leaf node which points to the segment that contains the key.

Since the B+ tree is used to index segments rather than individual points, the runtime for searching for the segment that a key belongs in is $O(log_b(p))$, where $b$ is a constant representing the fanout of the tree (i.e., number of separators inside a node) and $p$ is the number of segments created during the segmentation process.

\subsubsection{Segment Search}
Once \system{} finds the segment for a key, it then must find the key's position inside the segment.
Recall that segments are created such that an element is no more than a constant distance ($error$) from the element's position determined through linear interpolation.
Other techniques for interpolation search inside a fixed-size index page are discussed in ~\cite{interpolation_search}.

To compute the approximate location of a key $k$ within a given segment $s$, we subtract the key from the first key that appears in the segment $s.start$.
Then, we multiply the difference by the segment's slope $s.slope$, as shown below in the following equation.

\begin{equation}
    pred\_pos = (k - s.start) \times s.slope
\end{equation}

After interpolating an element's position, the true position of an element is guaranteed to be within the error threshold.
Therefore, \system{} locally searches the following region using binary search (as shown in Algorithm~\ref{algo:lookup}).

\begin{equation}
    true\_pos \in [pred\_pos - error, pred\_pos + error]
\end{equation}

However, it is also important to note that any search algorithm, including linear search, binary search, or exponential search can also be used depending on the specific scenario (e.g., hardware properties, error threshold).


Since segments satisfy the specified error condition, the cost of searching for an element inside a segment is bounded.
More specifically, the runtime for locating an element inside a segment is $O(log_2(error))$ where $error$ is constant.


\subsection{Range Queries}

Range queries, unlike point queries, have the additional requirement that they examine every item in the specified range.
Therefore, for range queries, the selectivity of the query (i.e., number of tuples that satisfy the predicate) has a large influence on the total runtime of the query.

However, like point queries, range queries must also find a single tuple: either the start or the end of the range.
Therefore, \system{} uses the previously described point lookup techniques to find the beginning of the specified range.
Then, since segments either store keys contiguously (clustered index) or have an indirection layer with pointers that is sorted by the key (non-clustered index), \system{} can simply scan from the starting location until it finds a key that is outside of the specified range.
For a clustered index, scanning the relevant range performs only very efficient sequential access, while for a non-clustered index, range queries require random memory accesses (which is true for any non-clustered index).

\begin{algorithm}[t]
\scriptsize
\begin{codebox}
\Procname{$\proc{Lookup}\mathit{(tree,key)}$}
\li $seg \gets \func{SearchTree}\mathit{(tree.root,key)}$
\li $val \gets \func{SearchSegment}\mathit{(seg,key)}$
\li \Return $val$
\End
\end{codebox}
\begin{codebox}
\Procname{$\proc{SearchTree}(node,key)$}
\li $i \gets 0$
\li \While $key < node.keys[i]$ 
\li \Do
$i \gets i+1$
\End
\li \If $node.value[i].isLeaf()$
\li \Then
$j \gets 0$
\li \While $key < node.values[j]$
\li \Do
$j \gets j+1$
\End
\li \Return $node.values[j]$
\End
\li \Return \func{SearchTree}$(node.values[i],key)$
\End
\end{codebox}
\begin{codebox}
\Procname{$\proc{SearchSegment}(seg,key)$}
\li $\mathit{pos \gets (key-seg.start) \times seg.slope}$
\li \Return \func{BinarySearch}$\mathit{(seg.data, pos-error, pos+error}, key)$
\End
\end{codebox}
\caption{Lookup Algorithm}
\label{algo:lookup}
\end{algorithm}
\section{Index Inserts}
\label{sec:inserts}
Along with locating an element, an index needs to be able to handle insert operations.
In some applications, maintaining a strict ordering guarantee is necessary, in which case \system{} should ensure that new items are inserted in-place.
However, in situations where this is not the case, we've developed a more efficient insert strategy that improves insert throughput.
In the following, we discuss each of these strategies for inserting new items into a \system{}.
Then, in Section~\ref{sec:experiments:inserts}, we show how these strategies compare for various workloads and parameters.


\subsection{In-place Insert Strategy}
\label{sec:inserts:inplace}

In a typical B+ tree that uses paging, pages are left partially filled and new values are inserted into an empty slot using an in-place strategy.
When a given page is full, the node is split into two nodes, and the changes are propagated up the tree (i.e., the inner nodes in the tree are updated).



Although similar, insert operations in \system{} require additional consideration since any key in the segment must be no more than specified error amount ($error$) away from its interpolated position.
Importantly, in-place inserts require moving keys in the page to preserve the order of the keys.

Without any a priori knowledge about the error of a given key, any attempt to move the key requires checking to see if the error condition is satisfied.
To make matters worse, a single insert may require moving many keys (in the worst case, all keys in the page) to maintain the sorted order. 
Thus, we must have a priori knowledge about any given key to determine if it can be moved in any direction while preserving the error guarantee.

Similar to the fill factor of a page, we divide the specified  $error$ in 2 parts: the segmentation error $e$ (error used to segment the data), and an insert budget $\varepsilon$ (number of locations a key can be moved in any direction). 
To preserve the specified error, we require that $error = e + \varepsilon$.
By keeping an insert budget for each page, \system{} can ensure that inserting a new element will not violate the error for the page.

More specifically, given a segment $s$, the page has a total size of $|s| + 2\varepsilon$ ($|s|$ is the number of locations in the segment).
Data is placed in the middle of the new page, yielding $\varepsilon$ empty locations at the beginning and end of the page.

With this strategy, it is possible to move any key in a direction which has free space without violating the error condition.
Therefore, to insert a new item using an in-place insert strategy, \system{} first locates the position in the page where the new item belongs.
Then, depending on which end of the page (left or right) is closer, all elements are shifted (either left or right) into the empty region of the page.
Once all of the empty spaces are filled, the segment needs to be re-approximated (using the segmentation algorithm described in Section~\ref{sec:bulk_loading}).
If the segmentation algorithm produces more than one segment, we create $n$ new segments (where $n$ is the number of segments produced after running Algorithm~\ref{alg:streaming_alg} on the single segment that is now full).
Finally, each new segment is inserted into the upper level tree, and any references to the old segment are deleted.

\subsection{Delta Insert Strategy}
\label{sec:inserts:delta}
Since the segments in \system{} are of variable size, the cost of an insert operation using the previously described in-place insertion strategy can be high, particularly for large error thresholds or uniform data that produces large segments with many data items.
Specifically, on average, $\frac{|s|}{2}$ keys may need to be moved for a single insert operation, where $|s|$ is the number of locations in the segment.
Therefore, a better approach for inserting new data items into a \system{} should amortize the cost of moving keys in the segment.

\begin{algorithm}[t]
\scriptsize
\begin{codebox}
\Procname{$\proc{InsertKey}\mathit{(tree,key)}$}
\li $\mathit{seg} \gets \func{SearchTree}\mathit{(tree.root,key)}$
\li $\mathit{seg.buffer.insert(key)}$
\li \If $\mathit{seg.buffer.isFull()}$
\li \Then
\li $\mathit{segs = \proc{segmentation}(seg.data, seg.buffer)}$
\li \For $s \in \mathit{segs}$
\li \Do
    $\mathit{tree.insert(s)}$
\End
\li $\mathit{tree.remove(seg)}$
\End
\li \Return
\End
\end{codebox}
\caption{Delta Insert Algorithm}
\label{algo:insert}
\end{algorithm}

To reduce the overhead of moving data items inside a page when inserting a new item, each segment in a \system{} contains an additional fixed-size buffer instead of extra space at each end.
More specifically, as shown in Algorithm~\ref{algo:insert}, new keys are added to the buffer portion of the segment for which the key belongs to (line 2).
This buffer is kept sorted to enable efficient search and merge operations.

Once the buffer reaches its predetermined size ($\mathit{buff})$, it is combined with the data in the segment and then re-segmented using the previously described segmentation algorithm (Algorithm~\ref{alg:streaming_alg}) to create a series of valid segments that satisfy the error threshold (line 5).
Note that depending on the data, the number of segments after this process can be one (i.e., the data inserted into the buffer does not violate the error threshold) or several.
Finally, each of the new segments generated from the segmentation process are inserted into the tree (line 6-7) and the old page is removed (line 8).

Storing additional data inside a segment impacts how to locate a given item, as well as how the error is defined.
Since adding a buffer for each segment can violate the error guarantees that \system{} provides, we transparently incorporate the buffer's size into the error parameter for the segmentation process.
More formally, given a specified error of $error$, we transparently set the error threshold for the segmentation process to ($\mathit{error-buff}$).
This ensures that a lookup operation will satisfy the specified error even if the element is located in the buffer.

The overall runtime for inserting a new element into a \system{} is the time required to locate the segment and add the element to the segment's buffer.
With $p$ pages stored in a \system{}, and a fanout of $b$ (i.e., number of keys in each internal separator node), inserting a new key into a \system{} has the following runtime: 
\begin{equation}
    \mathit{insert \ runtime: O(log_bp) + O(buff)}
\end{equation}
Note that when the buffer is full and the segment needs to be re-segmented, the runtime has an additional cost of $O(d)$, where $d$ is the sum of a segment's data size and buffer size.
Additionally, if the write-rate is very high, we could also support merging algorithms that use a second buffer similar to how column stores merge a write-optimized delta to the main compressed column.
However, this is an orthogonal consideration that heavily depends on the read/write ratio of a workload and is outside the scope of this paper.

\section{Cost Model}
\label{sec:cost_model}

Since the specified error threshold affects both the performance of lookups and inserts as well as the index's size, the natural question follows: how should a DBA pick the error threshold for a given workload?
To navigate this tradeoff, we provide a cost model that helps a DBA pick a ``good'' error threshold when creating a \system{}.
At a high level, there are two main objectives that a DBA can optimize: performance (i.e., lookup latency) and space consumption~\cite{monkey,tradeoffs}.
Therefore, we present two ways to apply our cost model that help a DBA choose an error threshold.

\subsection{Latency Guarantee}
For a given workload, it is valuable to be able to provide latency guarantees to an application.
For example, an application may require that lookups take no more than a specified time threshold (e.g., $1000$ns) due to SLAs or application-specific requirements (e.g., for interactive applications).
Since \system{} incorporates an error term that in turn affects performance, we can model the index's latency in order to pick an error threshold that satisfies the specified latency requirement.

As discussed, lookups require finding the relevant segment and then searching the segment (data and buffer) for the element.
Since the error threshold influences the number of segments that are created (i.e., a smaller error threshold yields more segments), we use a function that returns the number of segments created for a given dataset and error threshold.
This function can either be learned for a specific dataset (i.e., segment the data using different error thresholds) or a general function can be used (e.g., make the simplifying assumption that the number of segments decreases linearly as the error increases).
We use $S_e$ to represent the number of resulting segments for given dataset using an error threshold of $e$.

Therefore, the total estimated lookup latency for an error threshold of $e$ can be modeled by the following expression, where $b$ is the tree's fanout, $\mathit{buff}$ is a segment's maximum buffer size, and $c$ is a constant representing the latency (in ns) of a cache miss on the given hardware (e.g., $50ns$).
Moreover, the cost function assumes binary search for the area that needs to be searched within a segment bounded by $e$ as well as searching the complete buffer.

\vspace{-0.3cm}
\begin{equation} \label{eq1}
\begin{split}
\textsc{latency}(e) = c \underbrace{
        \big[
        log_b(S_e)
        }_\text{Tree Search}
        +
      \underbrace{
        log_2(e)
        }_\text{Segment Search} 
  + \underbrace{
        \mathit{log_2(buff)}
        \big]}_\text{Buffer Search} 
\end{split}
\end{equation}

Setting $c$ to a constant value implies that all random memory accesses have a constant penalty but caching can often change the penalty for a random access.
In theory, instead of being a constant, $c$ could be a function that returns the penalty of a random access but we make the simplifying that $c$ is a constant.

Using this cost estimate, the index with the smallest storage footprint that satisfies the given latency requirement $L_{req}$ (in nanoseconds) is given by the following expression, where $E$ represents a set of possible error values (e.g., $E = \{10,100,1000\}$) and $\textsc{SIZE}$ is a function that returns the estimated size of an index (defined in the next section).

\vspace{-0.3cm}
\begin{equation}
e = \argmin_{\{e \in E\  \ | \ \textsc{LATENCY}(e) \ \le \ L_{req} \} } \big( \textsc{SIZE}(e) \big)
\end{equation}

In addition to modeling the latency for lookup operations, we can similarly model the latency for insert operations.
However, there are a few important differences.
First, inserts do not have to probe the segment.
Also, instead of searching a segment's buffer, inserts require adding the item to the buffer in sorted order.
Finally, we must also consider the cost associated with splitting a full segment.

\subsection{Space Budget}

Instead of specifying a lookup latency bound, a DBA can also give \system{} a storage budget to use.
In this case, the goal becomes to provide the highest performance (i.e., lowest latency for lookups and inserts) while not exceeding the specified storage budget.

More formally, we can estimate the size of a read-only clustered index (in bytes) for a given error threshold of $e$ using the following function, where again $S_e$ is the number of segments that are created for an error threshold of $e$, and $b$ is the fanout of the tree.


\vspace{-0.3cm}
\begin{equation}
\textsc{SIZE}(e) = \underbrace{
                    (S_e \cdot log_b(S_e) \cdot 16B)
                    }_\text{Tree}
                    +
                    \underbrace{
                    (S_e \cdot 24B)
                    }_\text{Segment}
\end{equation}

The first term is a pessimistic bound on the storage cost of the tree (leaf + internal nodes using 8 byte keys/pointers), while the second term represents the added metadata about each segment (i.e., each segment has a starting key, slope, and pointer to the underlying data, each 8 bytes).

Therefore, the smallest error threshold that satisfies a given storage budget $S_{req}$ (given in bytes) is given by the following expression where again $E$ represents a set of all possible error values (e.g., $E = \{10,100,1000\}$).

\vspace{-0.3cm}
\begin{equation}
e = \argmin_{\{e \in E\  \ | \ \textsc{SIZE}(e) \ \le \ S_{req} \} } \big( \textsc{Latency}(e) \big)
\end{equation}

As we show in Section~\ref{sec:experiments:cost_model}, our cost model can accurately estimate the size of a \system{} over real-world datasets, providing DBAs with a valuable way to balance performance (i.e., latency) with the storage footprint of a \system{}.


\section{Evaluation}
\label{sec:experiments}

\begin{figure*}[!ht]
  \begin{minipage}[b]{.32\textwidth}
    \centering
    \includegraphics[width=.96\linewidth]{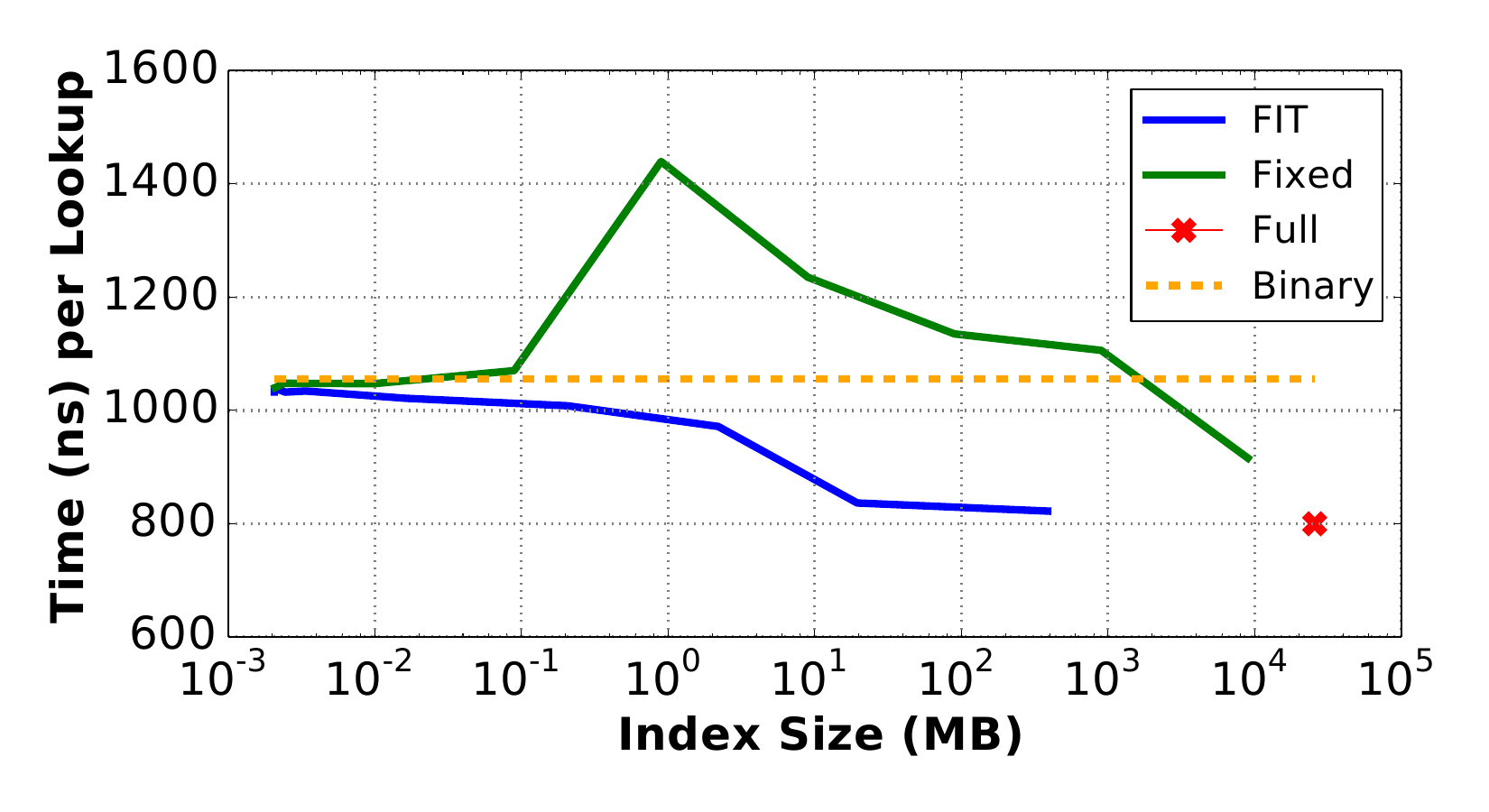}
    \vspace{-2mm}
    \subcaption{Weblogs}
    \label{fig:weblogs_lookup}
  \end{minipage}
  \hfill
  \begin{minipage}[b]{.32\textwidth}
    \centering
    \includegraphics[width=.96\linewidth]{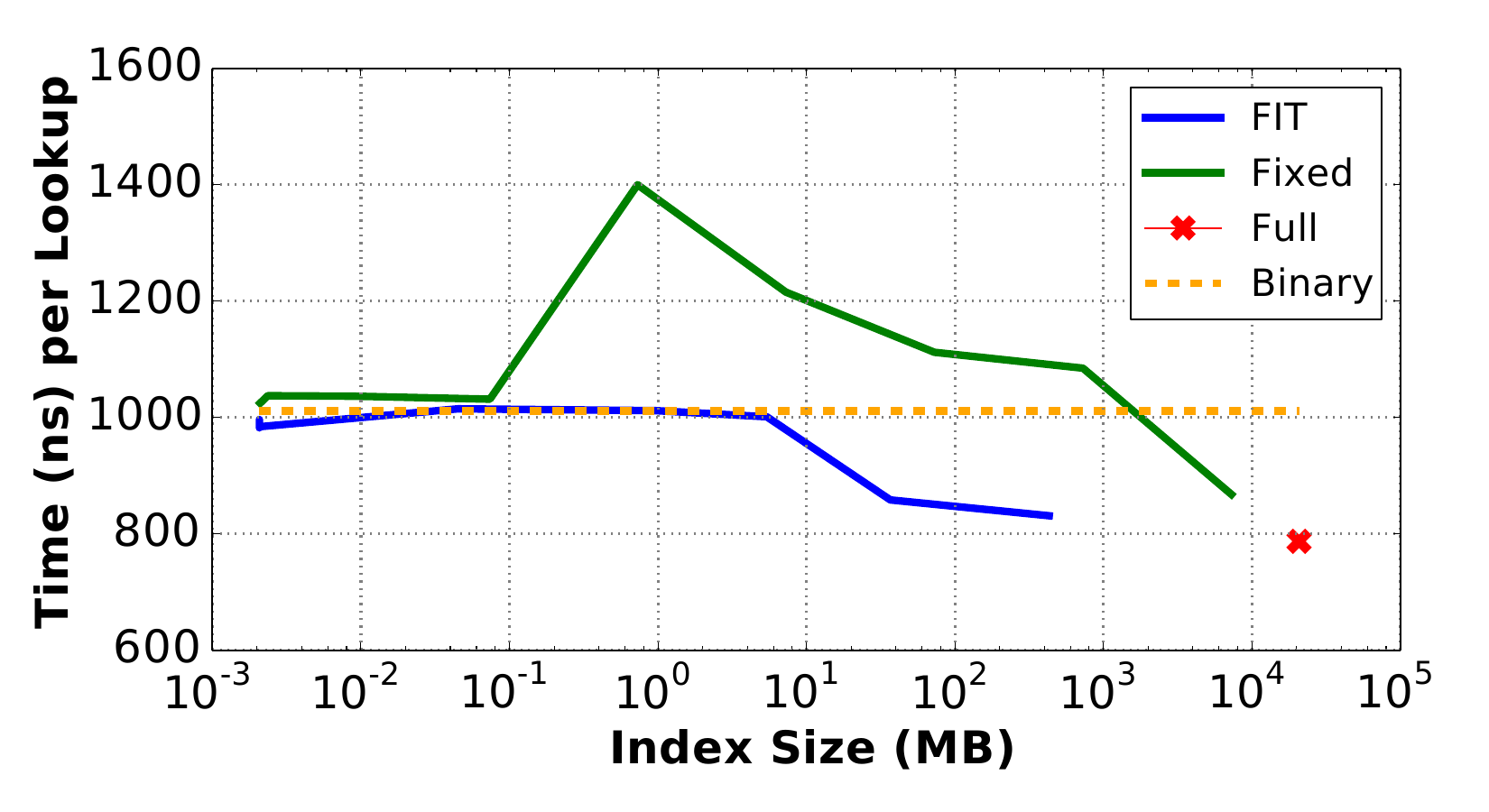}
    \vspace{-2mm}
    \subcaption{IoT}
    \label{fig:iot_lookup}
  \end{minipage}
  \hfill
  \begin{minipage}[b]{.32\textwidth}
    \centering
    \includegraphics[width=.96\linewidth]{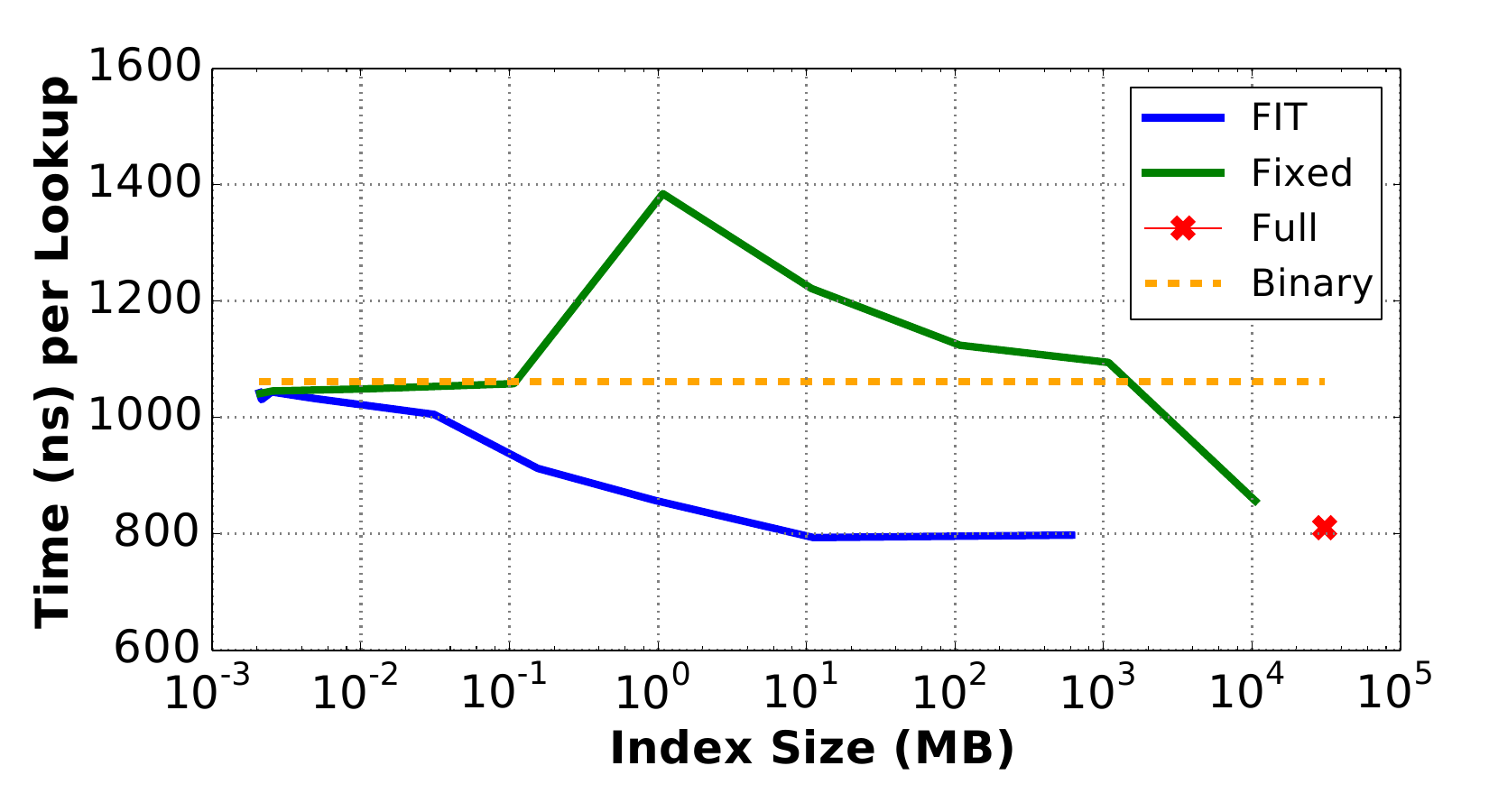}
    \vspace{-2mm}
    \subcaption{Maps}
  \label{fig:maps_lookup}
  \end{minipage}
  \hfill
  \vspace{-3mm}
  \caption{Latency for Lookups (per thread)}
  \vspace{-3mm}
  \label{fig:lookups}
\end{figure*}

\begin{figure*}[!ht]
  \begin{minipage}[b]{.32\textwidth}
    \centering
    \includegraphics[width=.9\linewidth]{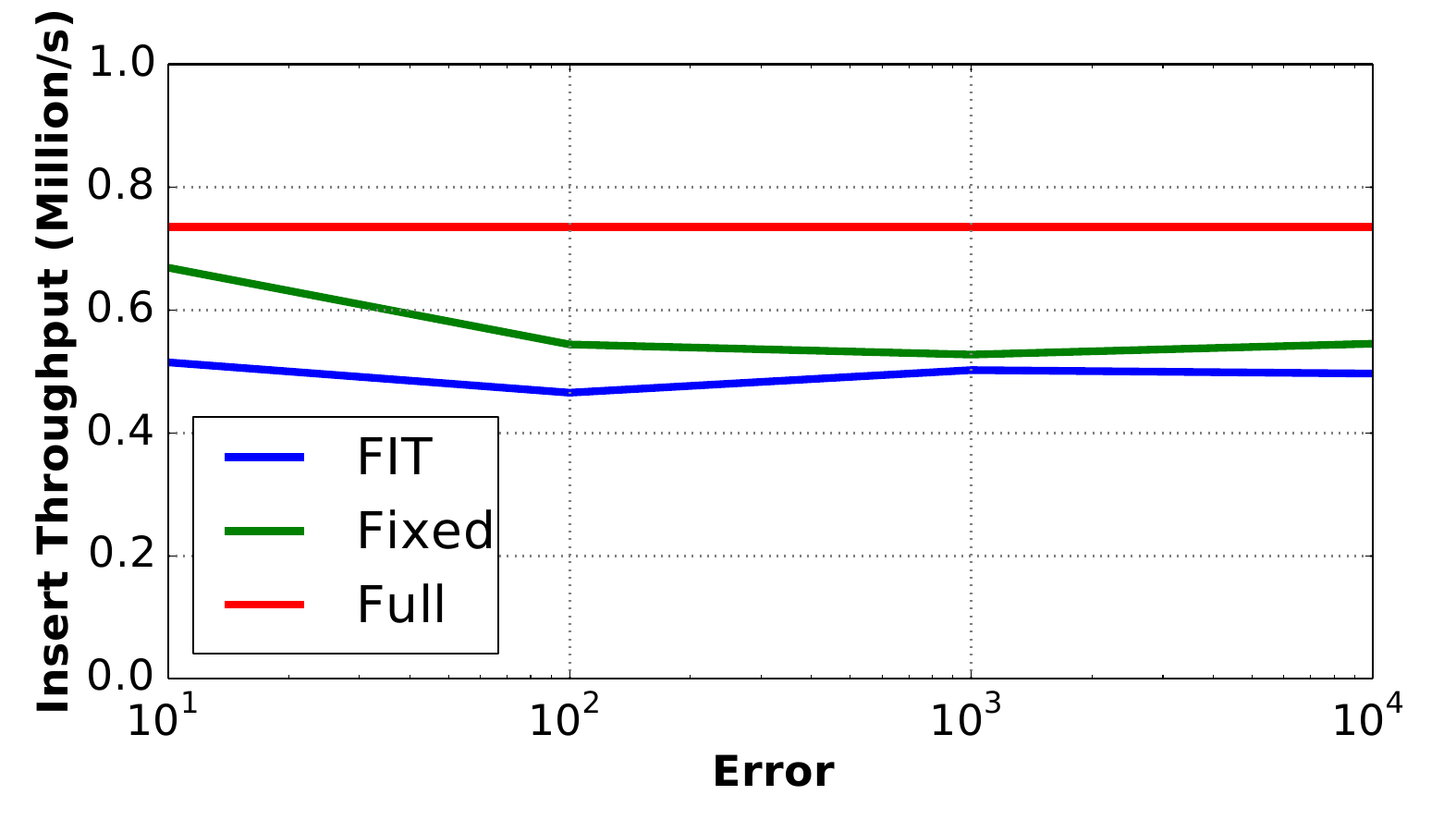}
    \vspace{-2mm}
    \subcaption{Weblogs}
    \label{fig:weblogs_insert}
  \end{minipage}
  \hfill
  \begin{minipage}[b]{.32\textwidth}
    \centering
    \includegraphics[width=.9\linewidth]{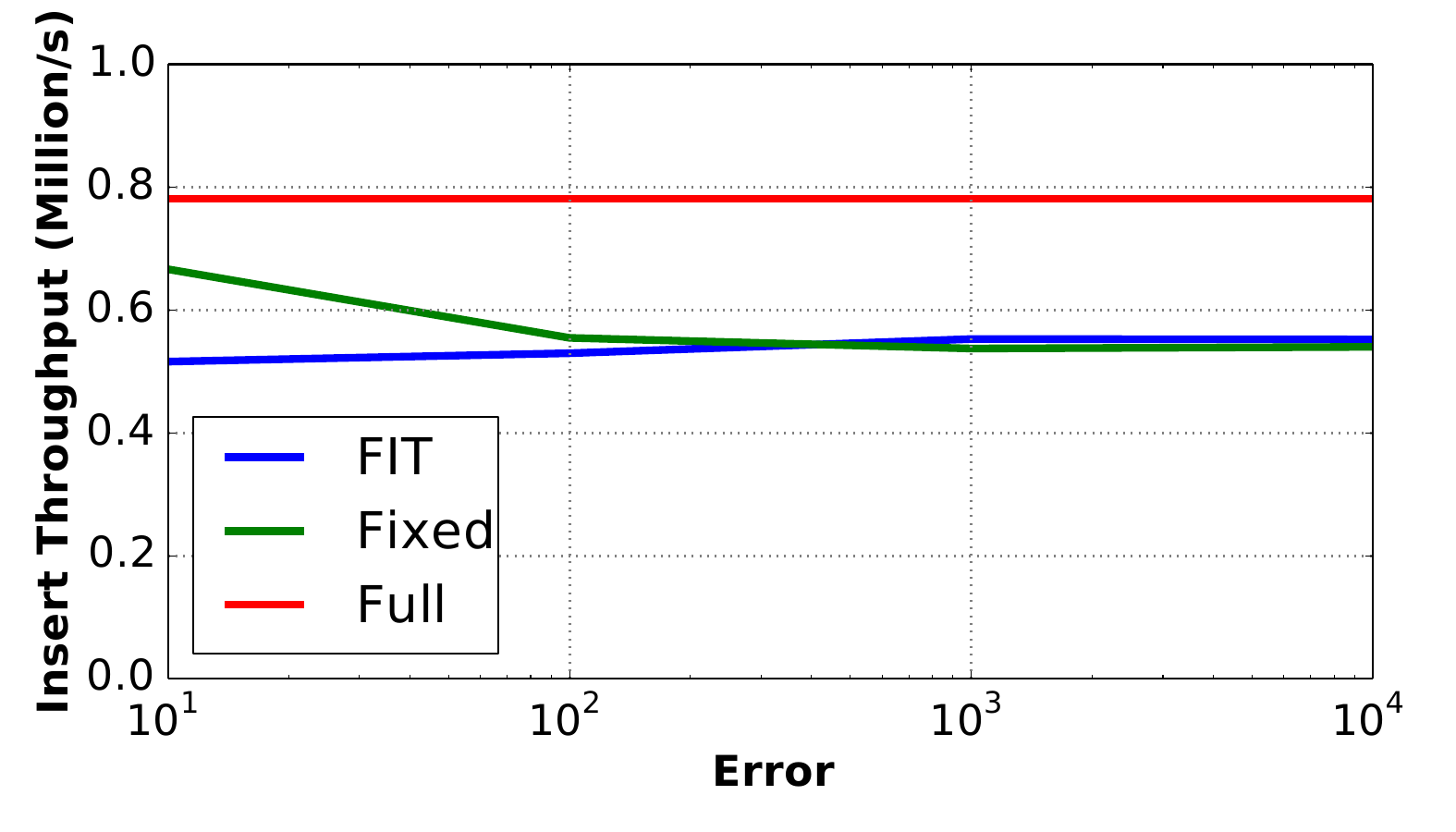}
    \vspace{-2mm}
    \subcaption{IoT}
    \label{fig:iot_insert}
  \end{minipage}
  \hfill
  \begin{minipage}[b]{.32\textwidth}
    \centering
    \includegraphics[width=.9\linewidth]{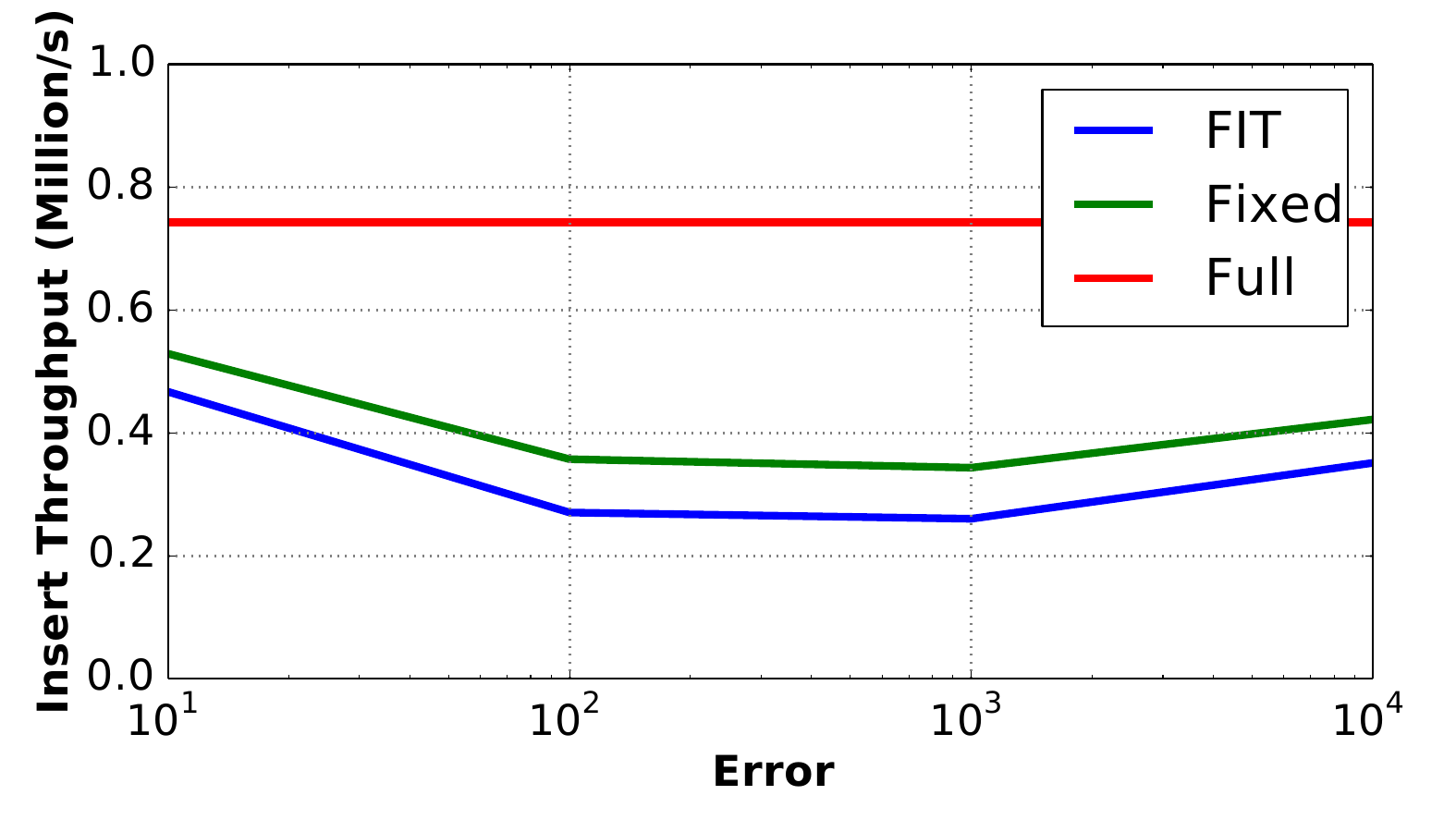}
    \vspace{-2mm}
    \subcaption{Maps}
  \label{fig:maps_insert}
  \end{minipage}
  \hfill
  \vspace{-3mm}

  \caption{Throughput for Inserts (per thread)}
  \label{fig:insert}
  \vspace{-3mm}
\end{figure*}

\begin{figure*}[!ht]
  \begin{minipage}[b]{.32\textwidth}
    \centering
    \includegraphics[width=.9\linewidth]{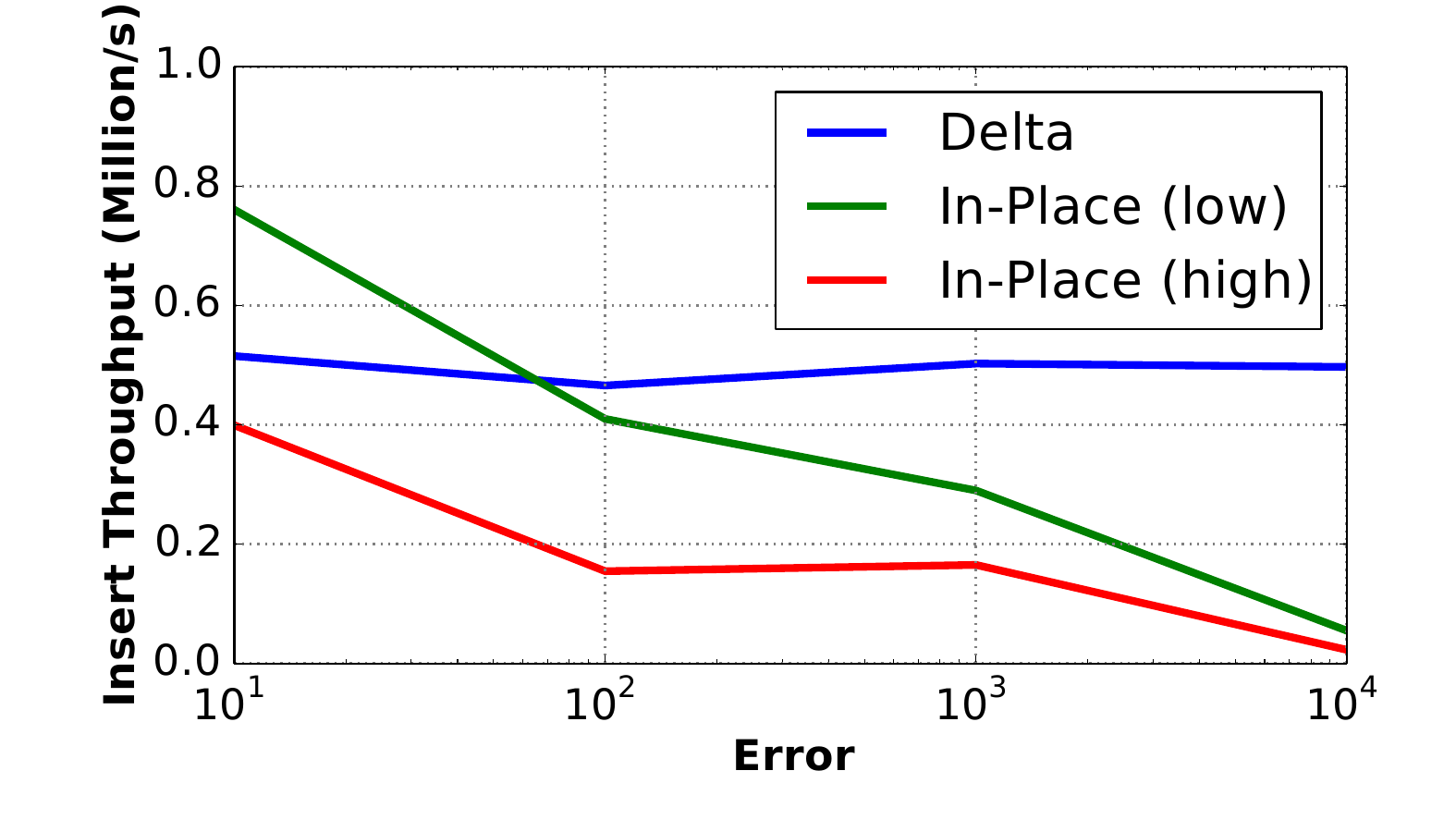}
    \vspace{-2mm}
    \subcaption{Weblogs}
    \label{fig:weblogs_insert_micro}
  \end{minipage}
  \hfill
  \begin{minipage}[b]{.32\textwidth}
    \centering
    \includegraphics[width=.9\linewidth]{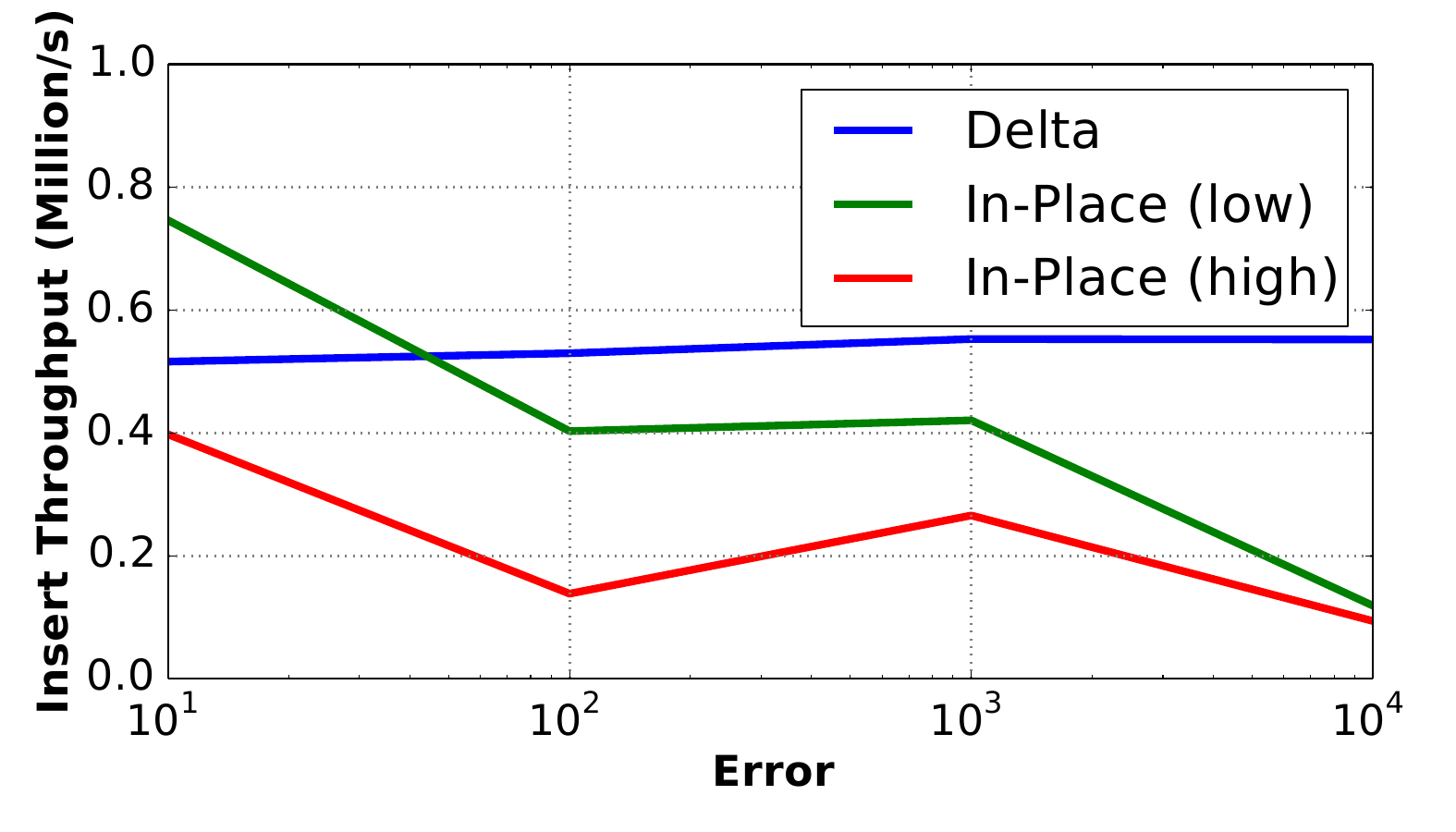}
    \vspace{-2mm}
    \subcaption{IoT}
    \label{fig:iot_insert_micro}
  \end{minipage}
  \hfill
  \begin{minipage}[b]{.32\textwidth}
    \centering
    \includegraphics[width=.9\linewidth]{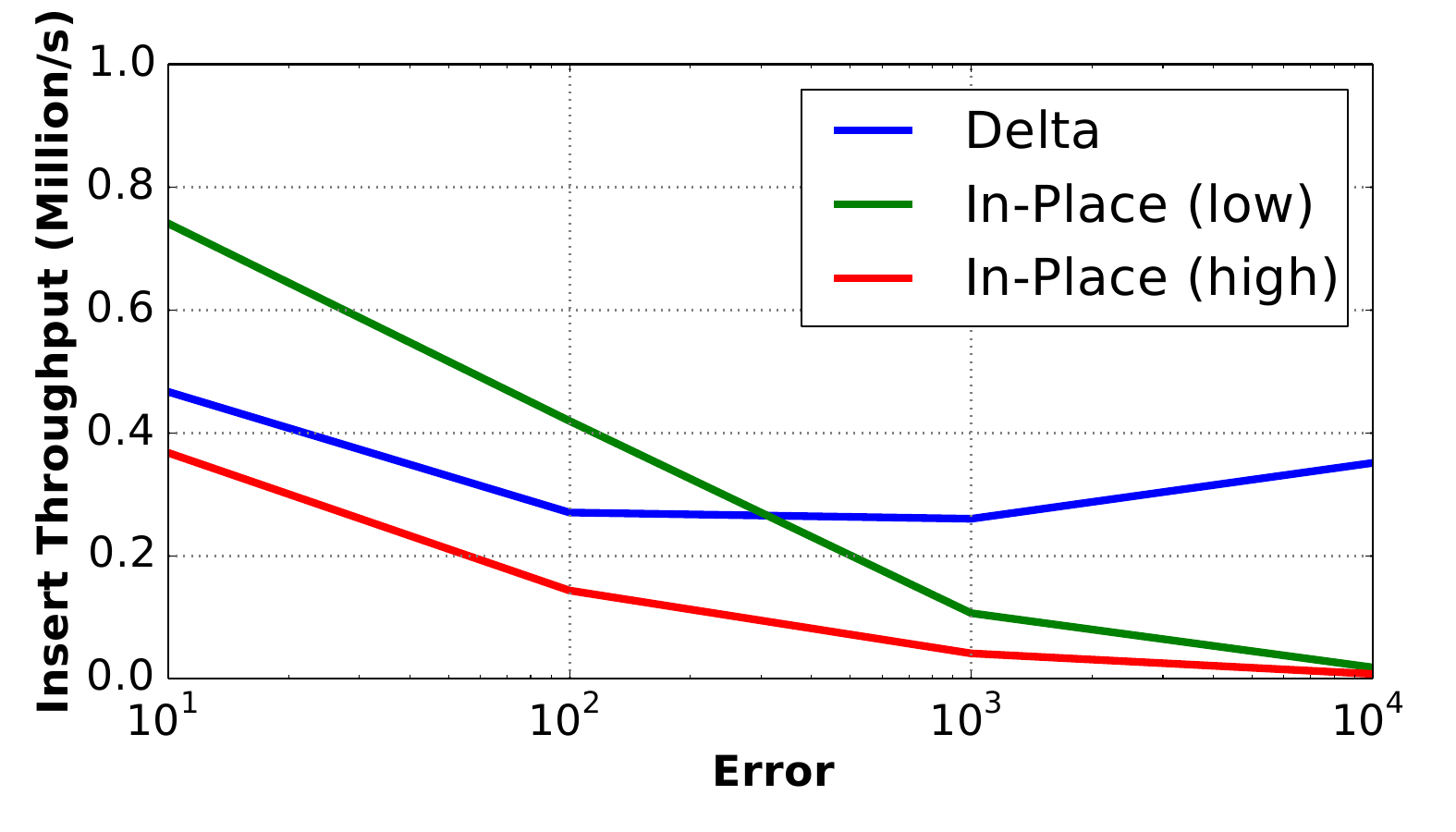}
    \vspace{-2mm}
    \subcaption{Maps}
  \label{fig:maps_insert_micro}
  \end{minipage}
  \hfill
    \vspace{-3mm}
  \caption{Insertion Strategy Microbenchmark}
  \label{fig:insert_micro}
  \vspace{0mm}
\end{figure*}

This section evaluates \system{} and the presented techniques.
Overall, we see that \system{} achieves comparable performance to a full index as well as indexes that use fixed-size paging while using orders of magnitude less space.

First, in Section \ref{sec:experiments:overall}, we compare the overall performance of \system{}, measuring its lookup and insert performance for both clustered and non-clustered indexes using a variety of real-world datasets.
Next, we compare the two proposed insert strategies in Section~\ref{sec:experiments:inplacedelta}.
Then, in Section~\ref{sec:experiments:index_construction}, we measure the construction cost of \system{} and in Section~\ref{sec:experiments:other_indexes} we show how \system{} can leverage other internal index structures.
Section~\ref{sec:experiments:scalability} shows the scalability of our index for different dataset sizes.
Finally, Section~\ref{sec:experiments:worst_case} shows how \system{} performs for an adversarial synthetically generated dataset (i.e., worst-case data distribution) and Section~\ref{sec:experiments:cost_model} evaluates our cost model.

Appendix~\ref{sec:appendix:further_eval} includes additional experiments that compare \system{} to Correlation Maps~\cite{correlation_maps}, show results for range queries, measure the throughput for various buffer sizes, and breakdown the lookup performance of \system{}.

We conducted all experiments on a single server with an Intel E5-2660 CPU (2.2GHz, 10 cores, 25MB L3 cache) and 256GB RAM and all index and table data was held in memory.

\subsection{Exp. 1: Overall Performance}
\label{sec:experiments:overall}
In the following, we evaluate the overall lookup and insert performance of \system{}. 
For these comparisons, we benchmark \system{} against both a full index (i.e., a dense index) as well as an index that uses fixed-size pages (i.e., a sparse index).
A full index can be seen as best case baseline for lookup performance and thus gives us an interesting reference point.

For the two baselines (full and fixed-size paging), we use a popular B+ tree implementation (STX-tree~\cite{stxtree}  v0.9) since our \system{} prototype also uses this tree to index the variable sized segments.
Importantly, as we show in Section~\ref{sec:experiments:other_indexes}, any other tree implementation can also serve as the organization layer.

\subsubsection{Datasets}
\label{sub:datasets}
Since performance of our index depends on the distribution of elements in a given dataset, we evaluate \system{} on real-world datasets with different distributions.
Later, in Section~\ref{sec:experiments:worst_case}, we show that our techniques are still valuable using a synthetically generated worst-case dataset. 
%
%
For our evaluation, we use three different real-world datasets, each with very different underlying data distributions: 
(1) Weblogs~\cite{mgbench}, 
(2) IoT~\cite{mgbench}, and 
(3) Maps~\cite{OSMdata}.

The Weblogs dataset contains $\approx 715M$ log entries for every web request to the CS department at a university over the past $14$ years.
This dataset contains subtle trends, such as the fact that more requests occur during certain times (e.g., school year vs. summer, day vs. night).
On the other hand, the IoT dataset contains $\approx 5M$ readings from around $100$ different IoT sensors (e.g., door, motion, power) installed throughout an academic building at a university.
Since these sensors generally reflect human activity, this dataset has interesting patterns, such as the fact there is more activity during certain hours because classes are in session.
For each of these datasets, we create a \emph{clustered} \system{} using the timestamp attribute (e.g., the time at which a resource was requested).
Finally, the Maps dataset contains the longitude of $\approx 2B$ features (e.g., museums, coffee shops) across the world.
Unsurprisingly, the longitude of locations is relatively linear and does not contain many periodic trends. 
Unlike the previous two datasets, we create a \emph{non-clustered} \system{} over the longitude attribute of this dataset.

For our approach, the most important aspect of a dataset that impacts \system{}'s performance is the data's periodicity.
For now, think of the periodicity as the distance between two ``bumps'' in a stepwise function that maps keys to storage locations as shown in Figure~\ref{fig:worst_case_data} (blue line).
If the specified error of \system{} is larger than the periodicity (green line), the segmentation results in a single segment.
However, if the error is smaller than the periodicity (red line), we need multiple segments to approximate the data distribution.

Therefore, we define a \textit{non-linearity ratio} to show the periodicity of a dataset.
To compute this ratio, we first calculate the number of segments required to cover the dataset for a given error threshold. 
We then normalize this result by the number of segments required for a dataset of the same size with periodicity equal to the error (which is the worst case, or the most ``non-linear'' in that scale).

To show that all datasets contain a distinct periodicity pattern, Figure~\ref{fig:linearity} plots the non-linearity ratio of each of dataset.
The IoT dataset has a very significant bump, signifying that there is very strong periodicity the scale of $10^4$, likely due to patterns that follow human behavior (e.g., day/night hours).
Weblogs has multiple bumps which are likely correlated to different periodic patterns (e.g., daily, weekly, and yearly patterns).
The Maps dataset, unlike the others, is linear at small scales (but has stronger periodicity at larger scales).

\subsubsection{Lookups}
The first series of benchmarks show how \system{} compares to (1) a full index (i.e., dense), (2) an index that uses fixed-size paging (i.e., sparse), and (3) binary search over the entire dataset.
We include binary search since it represents the most extreme case where the error is equal to the data size (i.e., our segmentation creates one segment).
For the Weblog and IoT dataset, we created a clustered index using the timestamp attribute which is the primary key of these datasets.
We created a non-clustered (i.e., secondary) index over the longitude attribute of the Maps dataset, which is not unique.


The results (Figure~\ref{fig:lookups}) show the lookup latency for various sizes of the index for the Weblogs (scaled to 1.5B records), IoT (scaled to 1.5B records), and Maps (not scaled, ~2B records) datasets.
More specifically, since the size of both \system{} and the fixed-size paging baseline can be varied (i.e., \system{}'s error term and the page size influence the number of leaf-level entries), we show the performance of each of these approaches for various index sizes.
Note that the size of a full index cannot be varied and is therefore a single point in the plot.
Additionally, since binary search does not have any additional storage requirement, its size is zero but is visualized as a dotted line.

In general, the results show that \system{} always performs better than an index that uses fixed-size paging.
Most importantly, however,  \system{} offers significant space savings compared to both fixed-size paging and a full index.
For example, in the Maps dataset, \system{} is able to match the performance of a full index using only $609MB$ of memory, while a full index consumes over $30GB$ of space.
Moreover, compared to a tree that uses fixed-size paging, a \system{} which consumes only $1MB$ of memory is able to match the performance of a fixed-size index which consumes over $10GB$ of memory, offering a space savings of four orders of magnitude.

As expected, for very small index sizes (i.e., very large pages or a high error threshold), both \system{} and fixed-size paging mimic the performance of binary search since there are only a few pages that are searched using binary search.
On the other hand, as the index grows, the performance of both fixed-size paging as well as \system{} converge to that of a full index due to the fact that pages contain very few elements.
Note that the spike in the graph for the fixed-size index is due to the fact that the index begins to fall out of the CPU's L2 cache.


Finally, as expected, the data distribution impacts the performance of \system{}.
More specifically, we can see that \system{} is able to more quickly match the performance of a full tree with the Maps dataset, compared to the Weblogs and IoT datasets.
This is due to the fact that the Maps dataset is relatively linear, when compared to the Weblogs and IoT datasets (shown in Figure~\ref{fig:linearity}).


\subsubsection{Inserts}
\label{sec:experiments:inserts}
Next, we compare the insert performance of \system{} to both a full index, as well as an index that uses fixed-size paging as previously described.
To insert new items, \system{} uses the previously described delta insert technique since it provides the best overall performance, which we show in next section where we perform an in-depth comparison of the delta and the in-place insert strategies.

More specifically, to ensure a fair comparison and that \system{} is not unfairly write-optimized, we set the size of the delta buffer to half of the specified error (i.e, for an error threshold of $100$, the underlying data is segmented using an error of $50$ and each segment's buffer has a maximum size of $50$ elements).
Similarly, for the index with fixed-size pages, the page size is given by the half of the error threshold we used for \system{} and the same amount (i.e., half of the error used in \system{}) is used as the buffer size.
As usual, once the buffer is full, the page is split into two pages.
The results, shown in Figure~\ref{fig:insert}, compare the throughput of each index after building the index over the first half of the data and inserting the second half for various error thresholds using shuffled versions of the previously described datasets.
As shown, \system{} is able to achieve insert performance that is, in general, comparable to an index that uses fixed-size paging.
Unsurprisingly, a full B+ tree is able to handle a higher write load than either a \system{} or an index that uses fixed-size paging since both need to periodically split pages that become full.
Additionally, \system{} needs to execute the segmentation algorithm, explaining the performance gap between \system{} and fixed-size paging.


\begin{figure}
  \vspace{-2mm}
  \begin{minipage}[b]{.51\columnwidth}
    \centering
    \includegraphics[width=0.99\linewidth]{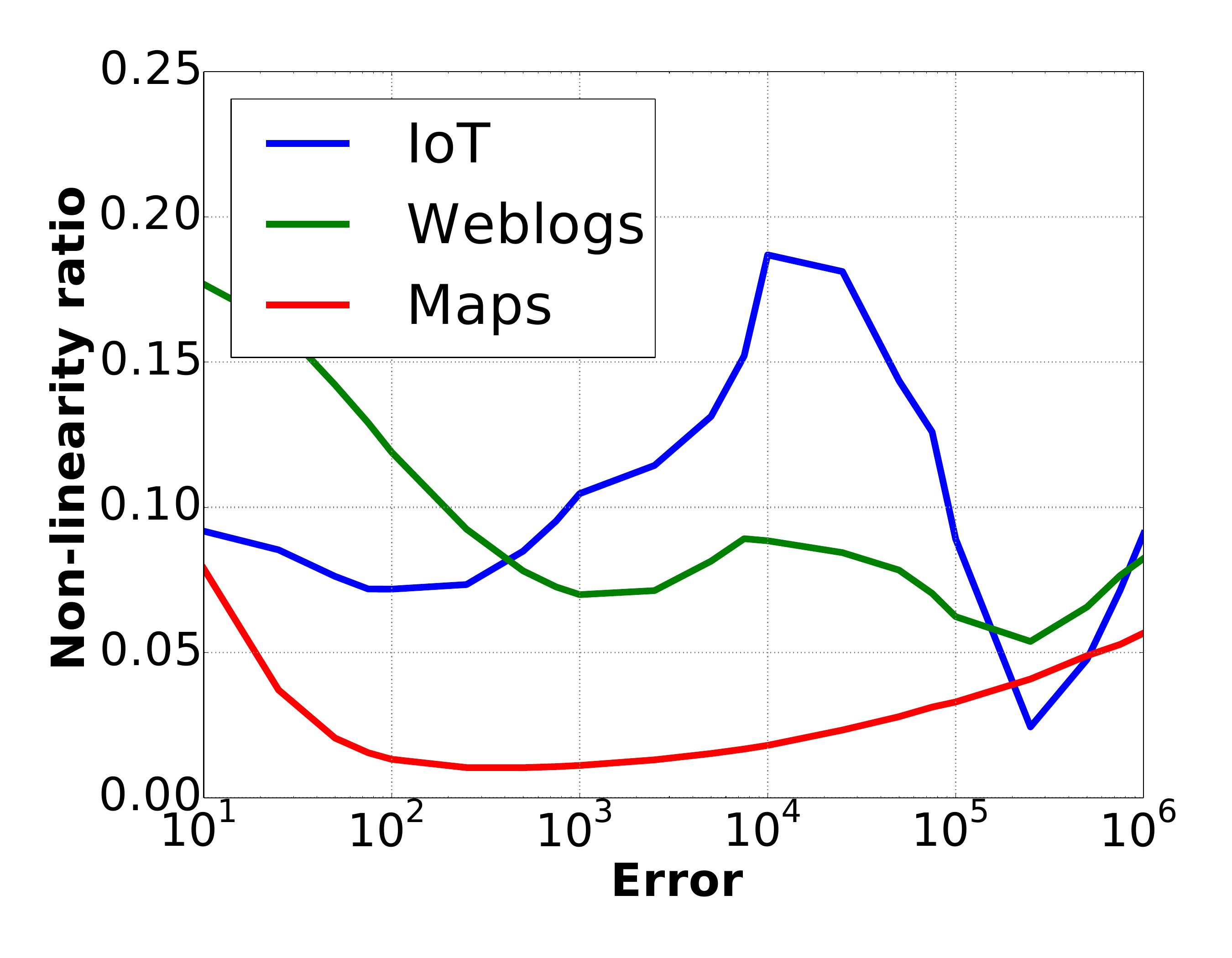}
    \vspace{-7mm}
    \caption{Non-linearity}
    \label{fig:linearity}
  \end{minipage}
  \hfill
  \begin{minipage}[b]{.48\columnwidth}
    \centering
    \includegraphics[width=\linewidth]{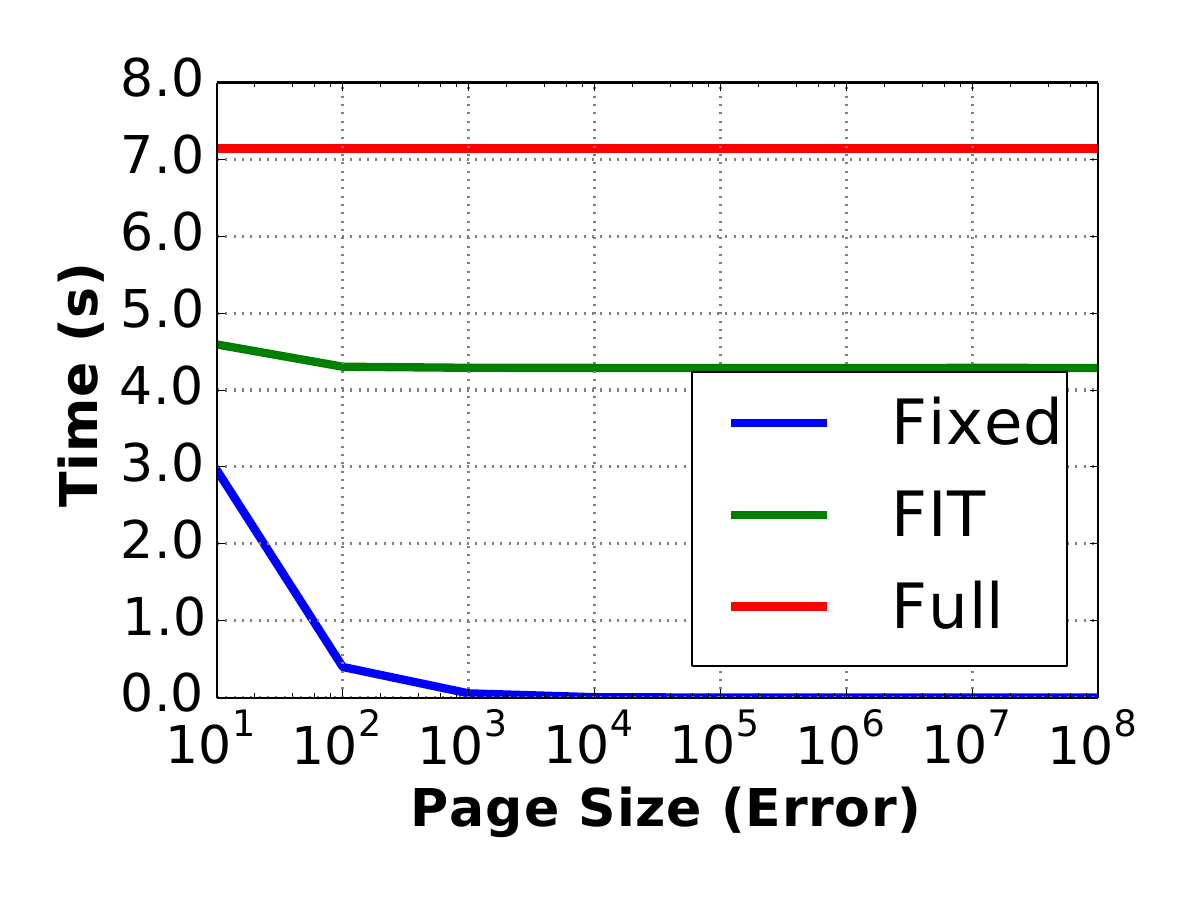}
    \vspace{-7mm}
    \caption{Build Time}
    \label{fig:weblogs_construction}
  \end{minipage}
  \hfill
  \vspace{-10mm}
\end{figure}

Interestingly, \system{} is faster than fixed-size paging in some cases since the error determines the number of segments in the tree.
For a large error, there typically are fewer segments generated, which reduces the number of times that \system{} needs to merge the buffer with the segment's data and execute the segmentation algorithm.

\subsection{Exp. 2: In-place vs. Delta Inserts}
\label{sec:experiments:inplacedelta}
In the following, we compare the two insert strategies described in Section~\ref{sec:inserts} and show how they perform for various datasets, fill factors, and error thresholds.

Figure~\ref{fig:insert_micro} shows the insert throughput for both the delta and in-place insert strategies for each of the previously described datasets.
As mentioned, for in-place inserts, the amount of free space reserved at the beginning and end of page ($\varepsilon$) can be tuned based on the read/write characteristics of the workload.
Therefore, we now show results for both a low (i.e., $\varepsilon$ is 25\% of the error) and high (i.e., $\varepsilon$ is 75\% of the error) fill factor.
For the delta insert results, we use the same setup as the previously described insert experiment (i.e., we set the size of the delta buffer to half of the specified error).

As shown, the delta insert strategy generally offers the highest insert throughput for error thresholds higher than $100$.
This is due to the fact that the in-place insert strategy must move many data items when inserting a new item since segments created with a higher error threshold contain more keys.
However, for low error thresholds, the in-place insert strategy outperforms the delta strategy, since there are significantly fewer data items that need to be shifted.

The fill factor impacts (1) how many data items are in a given segment and must be copied when inserting a new element, and (2) how often a segment fills up and needs to be re-segmented.
As shown, in general, the in-place strategy with a low fill factor offers the highest insert performance.

\begin{figure}
  \vspace{-4mm}
  \begin{minipage}[b]{.52\columnwidth}
    \centering
    \includegraphics[width=\linewidth]{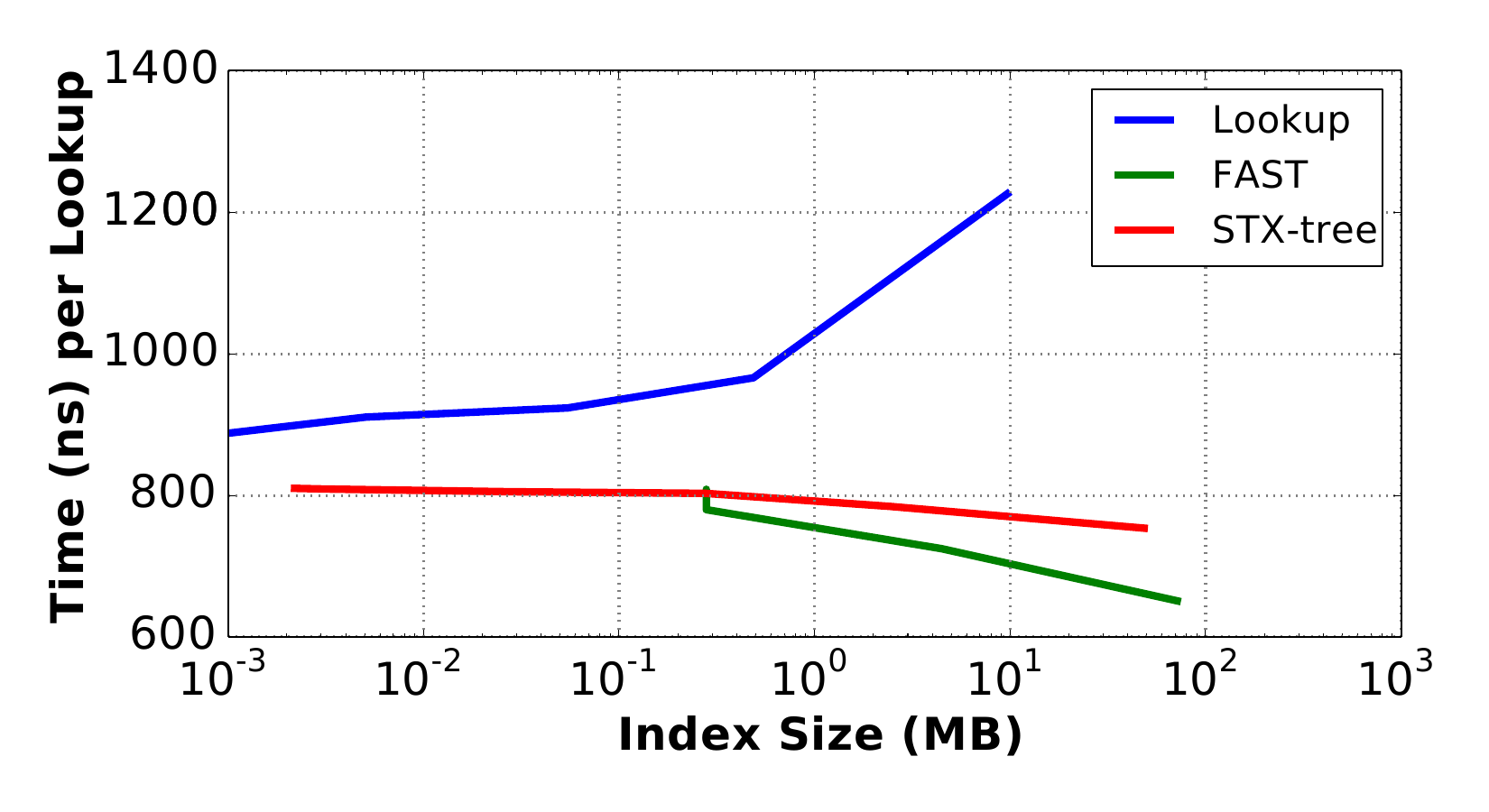}
    \vspace{-3mm}
    \caption{Other Indexes}
    \label{fig:generalizability}  \end{minipage}
  \hfill
  \begin{minipage}[b]{.47\columnwidth}
    \centering
    \includegraphics[width=0.97\linewidth]{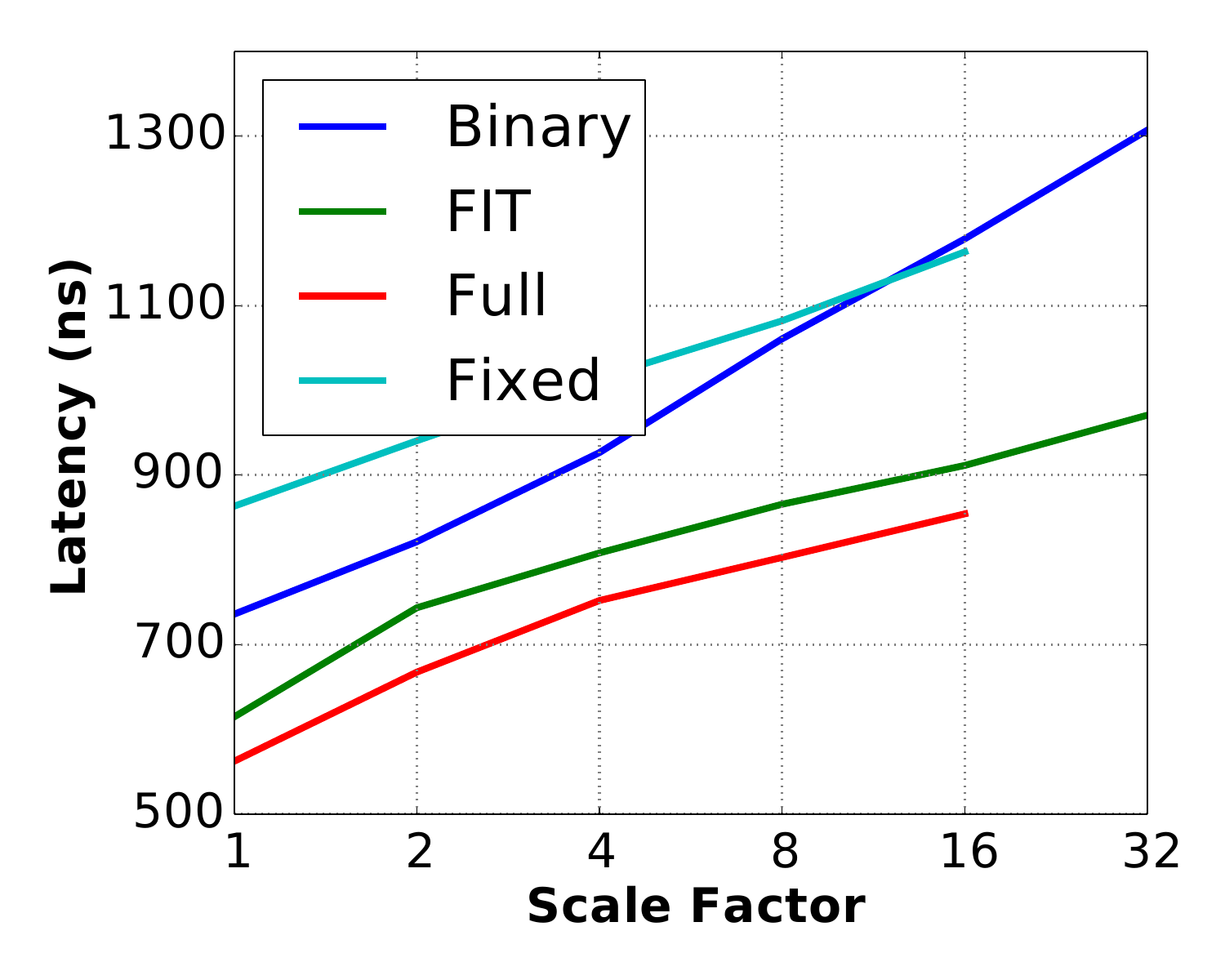}
    \caption{Scalability}
    \label{fig:weblogs_scalability}
  \end{minipage}
  \hfill
  \vspace{-8mm}
\end{figure}

\subsection{Exp 3: Index Construction}
\label{sec:experiments:index_construction}

In the following, we quantify the cost to construct a \system{}.
More specifically, we measure the amount of time required to bulk load a \system{} and compare it to the time required to construct a B+ tree that uses fixed-size pages as well as a full index.
The results, shown in Figure~\ref{fig:weblogs_construction}, plot the runtime for each approach for various page sizes (Fixed B+ tree) or error thresholds (\system{}) using the Weblogs dataset.

Since a full index must insert every element into the tree, it takes a constant amount of time.
However, the time required to bulk load a B+ tree that uses fixed-size pages decreases as the page size increases since only the page boundaries (e.g., every 100th element) must be accessed from the underlying data.
As previously mentioned, building a \system{} requires first segmenting the data, and then inserting each segment into the underlying tree structure.
Unsurprisingly, since the segmentation algorithm must examine every element in the data, a \system{} incurs a constant amount of extra overhead ($\approx$ 4.2 seconds in the shown experiment) compared to a B+ tree that uses fixed-size pages.

Interestingly, however, this constant amount extra overhead can be avoided in some cases.
For example, when initially loading the data into the system, it is possible to execute the segmentation algorithm in a streaming fashion.
In this case, it is possible that building a \system{} is actually less expensive than building a B+ tree with fixed-size pages since the resulting segmentation algorithm leverages trends in the data to produce fewer lead-level entries.

\subsection{Exp 4: Other Indexes}
\label{sec:experiments:other_indexes}
As mentioned, \system{} internally uses STX-tree~\cite{stxtree} to organize the variable sized pages generated through our segmentation process (Section~\ref{sec:bulk_loading}).
However, depending on the workload characteristics (e.g., read/write ratio), other internal data structures could also be used to index segments to provide additional performance improvements. 


Therefore, to show that our techniques are generalizable, we compare using an STX-tree to (1) FAST~\cite{FAST}, a highly optimized index structure for read-only workloads and (2) a simple lookup table that simply stores the variable sized pages in sorted order.
Figure~\ref{fig:generalizability} shows the results for various index sizes (i.e., error thresholds) over a subset of the Weblogs dataset (since FAST requires a power of two elements).
As shown, a lookup table provides superior space savings (since there are no internal separator nodes) but with a higher lookup latency.
On the other hand, a \system{} that internally uses the highly optimized FAST index can provide faster lookups but often consumes more space.
Therefore, a \system{} is able to leverage alternative indexes that may be more performant depending on the workload characteristics (e.g., read-only).

\begin{figure}
    \vspace{-4mm}
  \begin{minipage}[b]{.49\columnwidth}
    \centering
    \includegraphics[width=\linewidth]{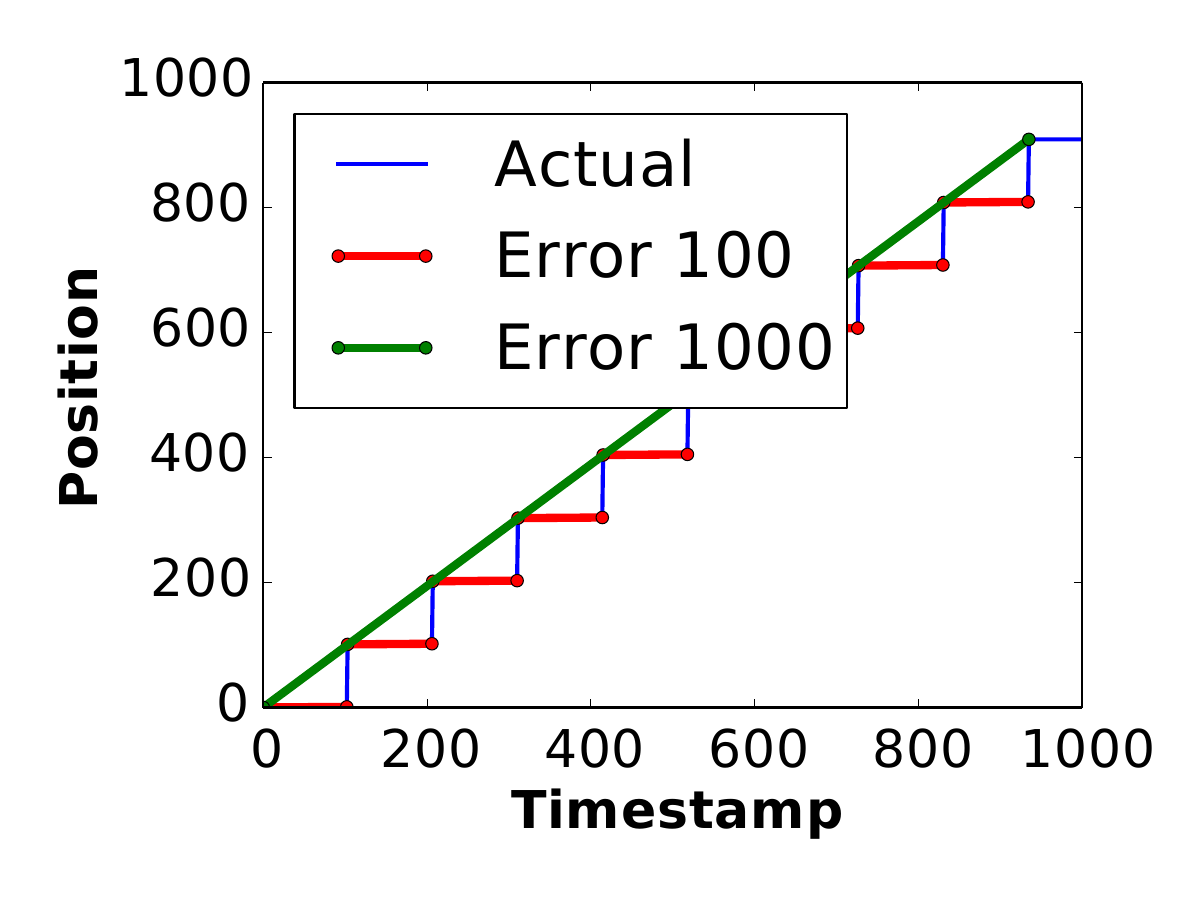}
    \subcaption{Worst Case Data}
    \label{fig:worst_case_data}
  \end{minipage}
  \hfill
  \begin{minipage}[b]{.49\columnwidth}
    \centering
    \includegraphics[width=\linewidth]{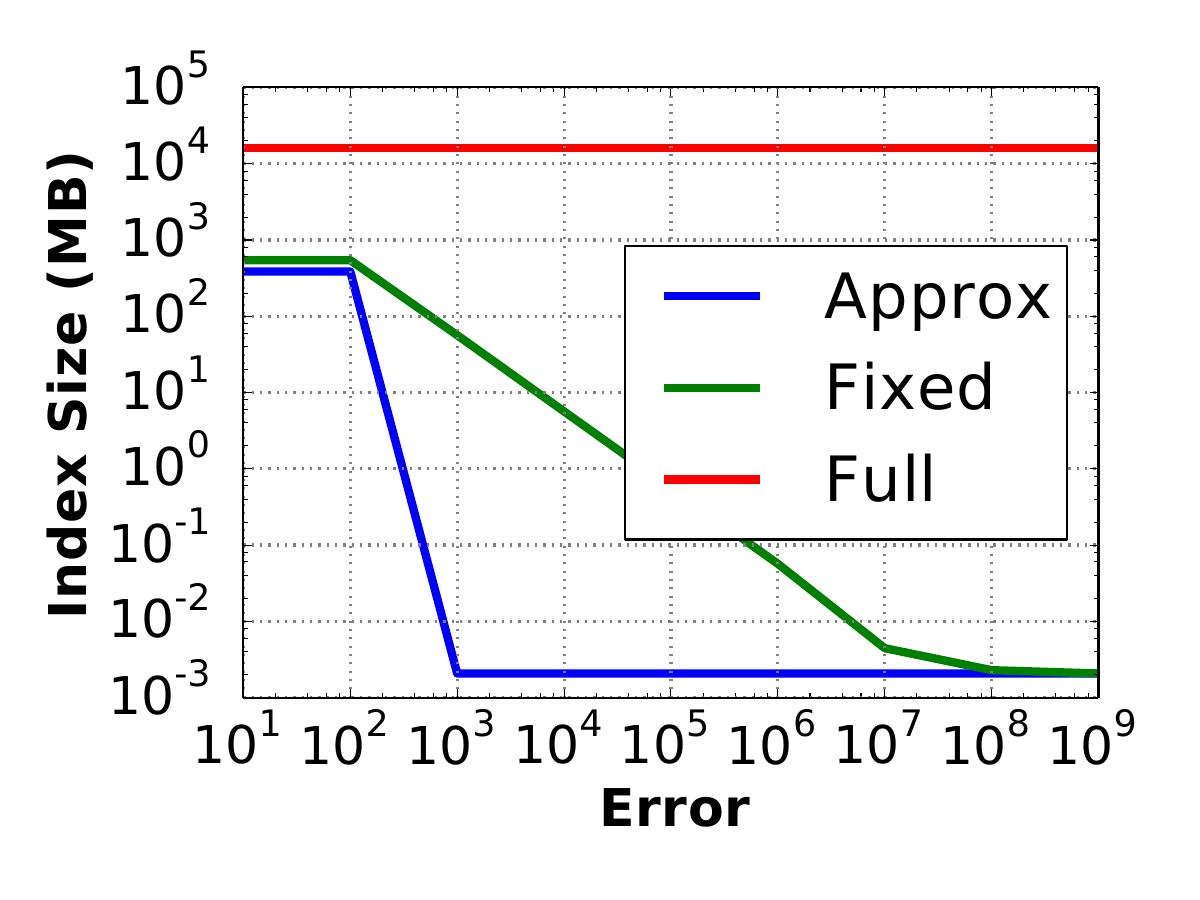}
    \subcaption{Worst Case Index Size}
    \label{fig:worst_case_size}
  \end{minipage}
  \hfill
  \caption{Worst Case Analysis}
  \vspace{-4mm}
  \label{fig:worst_casee}
\end{figure}

\subsection{Exp. 5: Data Size Scalability}
\label{sec:experiments:scalability}

To evaluate how \system{} performs for various dataset sizes, we measure the lookup latency for the Weblogs dataset for various scale factors where both the error threshold and fixed page size are set to $100$, which is optimal for this dataset.
Since the performance of our index depends on the underlying data distribution, we scale the dataset while maintaining the underlying trends.
We omit the result for all other datasets here (IoT and Maps) since they follow similar trends.

Figure~\ref{fig:weblogs_scalability} shows that the indexes (i.e., \system{}, a full index, fixed-size paging) scale better than binary search due to the better theoretical asymptotic runtime ($log_b(n)$ vs. $log_2(n)$) and cache performance.
Additionally, \system{}'s performance over various dataset sizes closely follows that of a full index which offers the best performance, demonstrating that our techniques offer valuable space savings without sacrificing performance.
Most importantly, neither a full index nor an index that uses fixed-size paging could scale to a scale factor of $32$ since the size of the index exceeded the amount of available memory.
This again shows that \system{} is able to offer valuable space savings.

\subsection{Exp 6: Worst Case Data}
\label{sec:experiments:worst_case}

Since the data distribution influences the performance of \system{}, we synthetically generated data to illustrate how our index performs with data that represents a worst-case.
We define the worst case as as a dataset which maximizes the number of segments given a specific error, described further in Section~\ref{sec:bulk_loading:segment_defi}.
To do this, we generate data using a step function with a fixed step size of $100$, as shown in Figure~\ref{fig:worst_case_data}.
Since the step size is fixed, an error threshold less than the step size yield a single segment per step.
However, given an error threshold larger than the step size, our segmentation algorithm will be able to use a single segment to represent the entire dataset.

Figure~\ref{fig:worst_case_size} shows the performance for various sizes of each index built over this worst case dataset.
As shown, for error thresholds of less than $100$, the size of a \system{} is the same as a fixed-size index but still smaller than a full index.
This is due to the fact that for the error thresholds less than the step size, \system{} creates segments of size 100 (step size), resulting in a large number of nodes in the tree.
On the other hand, for an error threshold that is larger than $100$, \system{} is able to represent the step dataset with only a single segment, dramatically reducing the index's size.
As shown, importantly, a \system{} will not contain more leaf-level entries than an index that uses fixed-size paging, as discussed in Section~\ref{sec:alg-analysis}.

\begin{figure}
    \vspace{-4mm}
  \begin{minipage}[b]{.49\columnwidth}
    \centering
    \includegraphics[width=\linewidth]{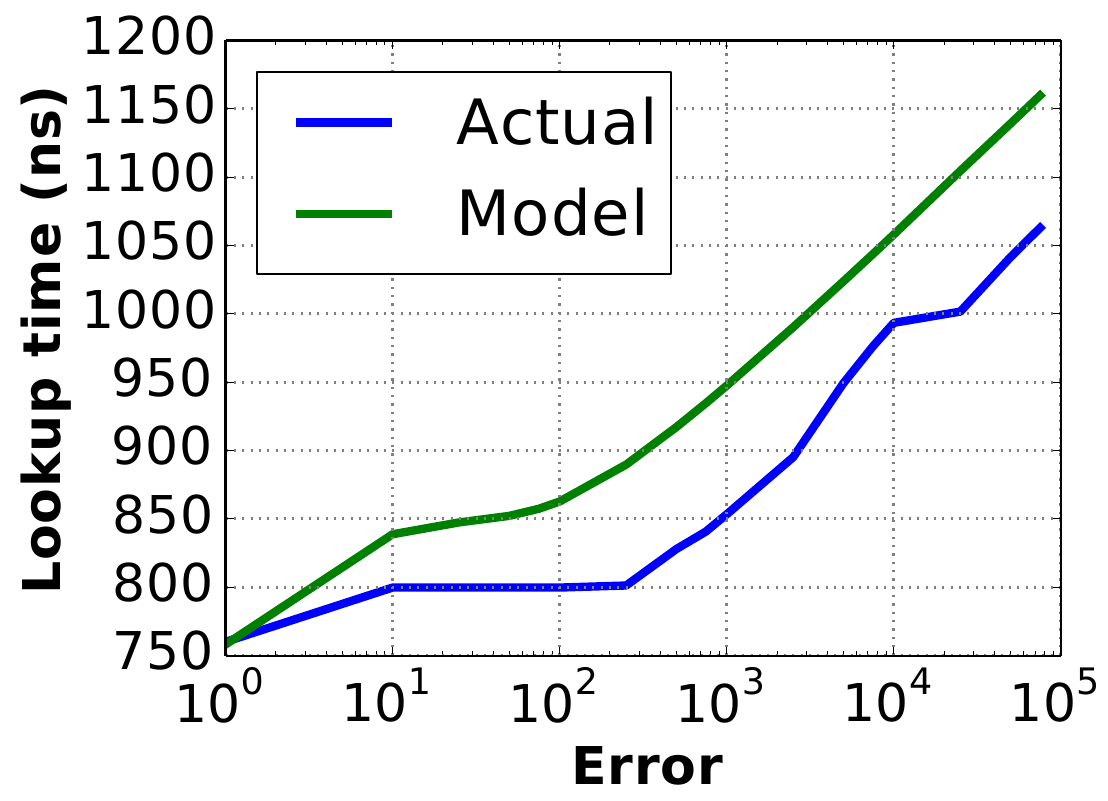}
    \subcaption{Latency}
    \label{fig:cost_model:latency}
  \end{minipage}
  \hfill
  \begin{minipage}[b]{.49\columnwidth}
    \centering
    \includegraphics[width=\linewidth]{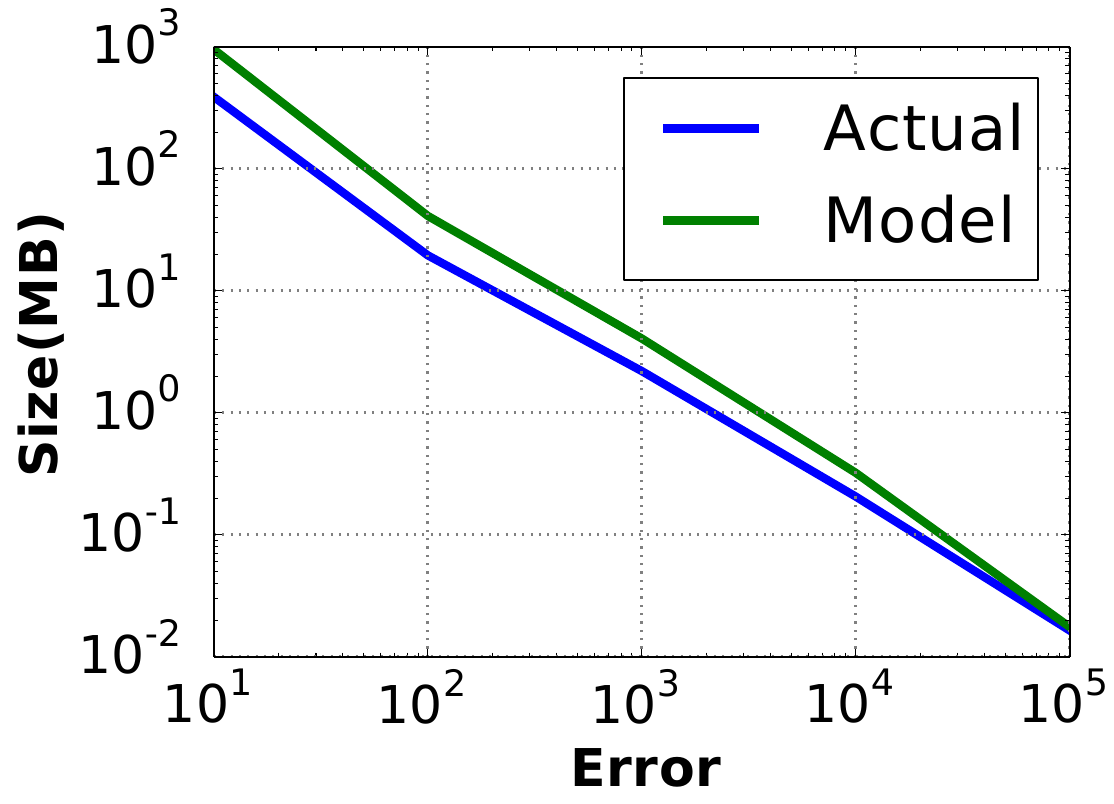}
    \subcaption{Size}
    \label{fig:cost_model:size}
  \end{minipage}
  \hfill
  \vspace{-1mm}
  \caption{Cost Model Accuracy}
  \vspace{-5mm}
  \label{fig:cost_model}
\end{figure}

\subsection{Exp. 7: Accuracy of Cost Model}
\label{sec:experiments:cost_model}
Since, as previously described, the error threshold influences both the latency as well as space consumption of our index, the cost model presented in Section~\ref{sec:cost_model} aims to guide a DBA when determining what error threshold to use for a \system{}.
More specifically, given a latency requirement (e.g., $1000ns$) or a space budget (e.g., $2GB$), our cost model automatically determines an appropriate error threshold that satisfies the given constraint.

Figure~\ref{fig:cost_model:latency} shows the estimated and actual lookup latency for various error thresholds on the Weblogs dataset using a value of $50ns$ for $c$ (the cost of a random memory access) determined through a memory benchmarking tool on the given hardware.
As shown, our latency model predicts an accurate upper bound for the actual latency of a lookup operation.
Our model slightly overestimates the latency due to the fact that it does not incorporate CPU caching effects. 
Since we overestimate the cost, we ensure that a specified latency threshold will always be observed.

To evaluate our size cost model, we show the predicted and actual size of a \system{} for various error thresholds in Figure~\ref{fig:cost_model:size}.
As shown, our model is able to accurately predict the size of an index for a given error threshold while ensuring that our estimates are pessimistic (i.e., the estimated cost is higher than the true cost).


\section{Related Work}
\label{sec:related}

The presented techniques in this paper overlap with work in different areas including (1) index compression, \linebreak (2) partial/adaptive indexes, and (3) function approximation.

\emph{Index Compression:}
Since B+ trees can often consume significant space, several index compression techniques have been proposed.
These approaches reduce the size of keys in internal nodes by applying techniques such as prefix/suffix truncation, dictionary compression, and key normalization~\cite{graefe, compressing_relations_and_indexes,rdf3x}.
Importantly, these techniques can also be applied within \system{} to further reduce the size of the underlying tree structure.

Similar to B+ tree compression, several methods have been proposed in order to more compactly represent bitmap indexes~\cite{bitmap1,bitmap2,bitmap3,bitmap4,bitmap5,bitmap_perf}.
Many of these techniques are specific to bitmap indexes, which are primarily only useful for attributes with few distinct values and not the general workloads that \system{} targets.

Correlation Maps~\cite{correlation_maps} try to leverage correlations between an unclustered attribute and a clustered attribute when an existing primary key index exists.
Our approach, on the other hand, does not assume an existing index already exists and uses variable sized paging (instead of fixed-sized buckets) that better model the underlying data.

FAST~\cite{FAST} is another more recent tree structure that organizes tree elements in a more compact and efficient representation in order to exploit modern hardware features (e.g., SIMD, cache line size) for read-heavy workloads.
Similarly, an Adaptive Radix Tree (ART)~\cite{art} leverages CPU caches for in-memory indexing.
Another idea discussed in~\cite{reducing_storage} are hybrid indexes which separate the index into hot and cold regions where cold data is stored in a compressed format.
Lastly, Log-structured Merge-trees~\cite{lsmt} are designed for mostly write intensive workloads and extensions including Monkey~\cite{monkey} balance performance and memory consumption.
Each of these techniques can be seen as orthogonal and thus also could be used by \system{} to more efficiently store the underlying tree structure as well as optimize for read-heavy workloads or hot/cold data.


Other indexing techniques have been proposed that store information about a region of the dataset, instead of the indexing individual keys. 
For example, leaf nodes in a BF-Tree~\cite{bftree} are bloom filters.
Unlike \system{}, BF-Tree does not exploit properties about the data's distribution when segmenting a dataset.
Another example are learned indexes~\cite{learnedindexes}, which aim to learn the underlying data distribution to index data items.
Unlike learned indexes, \system{} has strict error guarantees, supports insert operations, and provides a cost model to ensure predictable performance and size.

Sparse indexes like Hippo~\cite{hippo}, Block Range Indexes~\cite{postgres}, and Small Materialized Aggregates~\cite{sma} all store information about value ranges similar to the idea of segments in \system{}.
However, these techniques do not consider the underlying data distribution or bound lookup/insert latency.

Finally, several approximation techniques have been proposed in order to improve the performance of similarity search ~\cite{fast_similarity_search,string_similarity,similarity_index} (for string or multimedia data), unlike \system{} which uses approximation for compressing indexes optimized for traditional point and range queries.



\emph{Partial and Adaptive Indexes:}
Partial indexes~\cite{partialidx} aim to reduce the storage footprint of an index since they index only a subset of the data that is of interest to the user.
For example, Tail Indexes~\cite{idea,hilda_idea} store only rare data items in order to reduce the storage footprint of the overall index.
\system{}, on the other hand, supports queries over all attribute values but could be extended to index only ``important'' data ranges as well.
Furthermore, database cracking~\cite{dbcracking} is a technique that physically reorders values in a column in order to more efficiently support selection queries without needing to store secondary indexes.
Since database cracking reorganizes values based on past queries, it does not efficiently support ad-hoc queries, like \system{} can.

\emph{Function Approximation:}
The main idea of a \system{} is to approximate the data distribution using piece-wise linear functions and approximating curves using piece-wise functions is not new~\cite{ESCH196985,Braess1971,smartgrid,Liu:2008:NOM:1477069.1477485}.
The error metrics $E_2$ (integral square error) and $E_\infty$ (maximal error) for these approximations have been discussed as well as different segmentation algorithms~\cite{Pavlidis:1974:SPC:1311087.1311553, Elmeleegy:2009:OPL:1687627.1687645,neumann}.
Unlike prior work, we consider only monotonic increasing functions, $E_\infty$, and potentially disjoint linear segments.
Moreover, none of these techniques have been applied to indexing and therefore do not consider looking up or inserting data items.

More recent work~\cite{shatkay1996approximate,keogh2001online,FU2011164,simsearch,Xu:2012:AAO:2247596.2247620} specific to time series data also leverages piece-wise linear approximations to store patterns for similarity search.
While these approaches also trade-off the number of segments with the accuracy of the approximate representation, they do not aim to provide the lookup and space consumption guarantees that \system{} does, and do not have the analysis related to these guarantees.

Finally, other work~\cite{Ao:2011:EPL:2002974.2002975} leverages piece-wise linear functions to compress inverted lists by storing functions and the distances of elements from the extrapolated functions.
However, these approximations use linear regression (which minimizes $E_2$), and there are no bounds on the error.




\section{Conclusion}
\label{sec:conclusion}

In this paper, we introduced \system{}, a new index structure that incorporates a tunable error parameter to allow a DBA to balance lookup performance and space consumption of an index.
To navigate this tradeoff, we presented a cost model that determines an appropriate error parameter given either (1) a lookup latency requirement (e.g., $500ns$) or (2) a storage budget (e.g., $100MB$).
We evaluated \system{} using several real-world datasets and showed that our index can achieve comparable performance to a full index structure while reducing the storage footprint by orders of magnitude.
\section{Acknowledgements}
This research is funded in part by the NSF CAREER Awards IIS-1453171 and CNS-1452712, NSF Award IIS-1514491, Air Force YIP AWARD FA9550-15-1-0144, and the
Data Systems and AI Lab (DSAIL) at MIT, as well as gifts from Intel, Microsoft, and Google.
\clearpage

\begin{scriptsize}
\bibliographystyle{ACM-Reference-Format}
\bibliography{bib}
\end{scriptsize}

\begin{appendices}
\section{Segmentation Analysis}
In the following we provide additional information about our segmentation algorithm, \textsc{ShrinkingCone}, described in Section~\ref{sec:bulk_loading}.
First, we prove the minimum size of a segment produced by our algorithm.
Then, although efficient in practice, we show that our algorithm can be arbitrarily worse than an optimal algorithm when considering the number of segments it produces.

\subsection{\textsc{ShrinkingCone} Segment Size}
\label{prf:segment-size}

We prove the claim from Theorem \ref{thm:seg_size} regarding the size of a maximal linear segment.

\begin{proof}
Consider 3 arbitrary points $(x_1,y_1)$, $(x_2,y_2)$, $(x_3,y_3)$, where $x_1<x_2<x_3$ and $y_1<y_2<y_3$.
By definition, the linear function starts at the first point in a segment, and ends at the last point in the segment.
The linear segment is not valid if the distance on the $y$ axis ($loc$) is larger than the specified error.
Therefore, given the 3 points, a linear segment starting at $(x_1,y_1)$ and ending at $(x_3,y_3)$ is not feasible if:

\begin{equation}
y_2 - err > \frac{y_3 - y_1}{x_3 - x_1}\left( x_2 - x_1 \right) + y_1
\end{equation}

By rearranging the inequality we get:

\begin{flalign}
err &< y_2 - y_1 - \frac{y_3 - y_1}{x_3 - x_1}\left( x_2 - x_1 \right)\\ 
&= \left( y_3 - y_1 \right)\cdot \left( 1 - \frac{x_2 - x_1}{x_3 - x_1}\right) - \left( y_3 - y_2 \right) \label{eq:33}\\
&\leq \left( y_3 - y_1 \right)\cdot \left( 1 - \frac{x_2 - x_1}{x_3 - x_1}\right) - 1 \label{eq:34}
\end{flalign}

In \eqref{eq:33} $y_3$ was added and subtracted, and in \eqref{eq:34} we use the fact that $y_2$ and $y_3$ are integers (thus $y_3 - y_2 \geq 1$).
This provides a lower bound for the distance between the first point in a segment and the first point in the following segment:

\begin{align}
y_3 - y_1 &> \frac{err + 1}{1 - \frac{x_2 - x_1}{x_3 - x_1}} = \left(err + 1\right)\cdot\frac{x_3 - x_1}{x_3 - x_2} > err + 1\\
&\Rightarrow y_3 - y_1 > err + 1
\end{align}

Since $(x_3,y_3)$ is the first point outside of the segment, the number of locations in the segment is $y_3-1-y_1\geq err + 1$.
\end{proof}

\subsection{\textsc{ShrinkingCone} Competitive Analysis}

\label{prf:ca}

In the following, we prove that \textsc{ShrinkingCone} can be arbitrarily worse than the optimal solution when considering the number of segments produced (i.e., \textsc{ShrinkingCone} is not competitive).

\begin{proof}
Given the error threshold $E = 100$, consider the following input to \textsc{ShrinkingCone}:
\begin{enumerate}
    \item 
    3 keys $(x_1,y_1),(x_2,y_2),(x_3,y_3)$ where $y_1=1,y_2=2,y_3=3$ and $x_3-x_2=x_2-x_1 = \frac{E}{2}$ (this is step 1 in \autoref{fig:comp_analysis}).
    \item
    The key $x_4 = x_3 + \frac{1}{E}$ repeated $E+1$ times (using $E+1$ consecutive locations), and the key $x_5 = x_4 + \frac{1}{E}$ without repetitions (using 1 location).
    
    After that repeat for $i\in [1,N]$ the following pattern: the key $x_{2(i + 2)} = x_{2(i + 2) - 1} + E$ repeated $E + 1$ times, and a single appearance of the key $x_{2(i + 2) + 1} = x_{2(i + 2)} + \frac{1}{E}$ (this is step 2 in \autoref{fig:comp_analysis}).
    \item
    The key $x_{2(N + 1 + 2)} = x_{2(N + 1 + 2) - 1} + \frac{E}{2}$ (step 3 in \autoref{fig:comp_analysis}).
\end{enumerate}

The algorithm will then create the following segments (an illustration is shown in \autoref{fig:comp_analysis}): 
\begin{itemize}
    \item
    $[x_1, x_4]$ (with slope $\frac{3}{E + \frac{1}{E}}$): adding the key $x_5$ will result in the slope $\frac{3+E+1}{E + \frac{2}{E}}$ which will not satisfy the error requirement for $x_4$, $y_1 + \frac{3+E+1}{E + \frac{2}{E}}\cdot (x_4-x_1) - y_4 = 1 + \frac{3+E+1}{E + \frac{2}{E}}\cdot (E + \frac{1}{E}) - 4 = 100.98> E$.
    \item
    Each of the next segments will contain exactly two keys (where the first key appears once, and the second key appears $E+1$ times), since otherwise the error for the second key will be $\frac{1 + E + 1}{E + \frac{1}{E}}\cdot E - 1 = 100.98 > E$.
    Just like before, the $E+1$ repetitions of a single key will cause a violation of the error (due to the spacing between subsequent keys).
\end{itemize}
Therefore, the algorithm will create $N+2$ segments given this input.

On the other hand, the optimal algorithm will need only $2$ segments: the first segment is the first key, and the second segment covers the rest of the input since the line starting at the second key and ending at the last key is never further away than $E$ from any key, due to the construction of the input. 
The slope of the second segment will be $\frac{3 + (N+1)\cdot(E + 2)}{E + (N+1)\cdot (E + \frac{1}{E})}$, and the first key in the segment is about $\frac{E}{2}$ away on the $x$ axis from the first repeated key.
Since the repeated keys are spaced evenly (distance on the $x$ axis of $E + \frac{1}{E}$), the linear function will not violate the error threshold for any key.

Since $N$ can be arbitrarily large, the algorithm is not competitive.
\end{proof}

\begin{figure}
\begin{center}
\includegraphics[width=0.65\columnwidth]{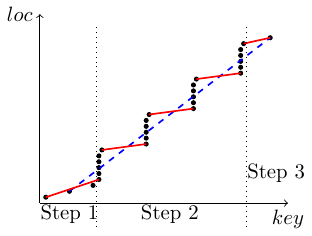}
\vspace{-2.5ex}
\caption{Competitive analysis sketch: the dots are the input, the dashed line is the optimal segmentation (the first dot is a segment), and the solid lines are the segments created by \textsc{ShrinkingCone}.}
\label{fig:comp_analysis}
\end{center}
\end{figure}

\section{Further Evaluation}
\label{sec:appendix:further_eval}
In this section, we present additional experimental results that provide a more in-depth study of \system{} and other approaches.
More specifically, we compare \system{} to Correlation Maps~\cite{correlation_maps}, show how \system{} performs for queries that involve range predicates, show how the buffer size impacts insert throughput, and finally breakdown the lookup latency.

\begin{figure}
  \begin{minipage}[b]{.47\columnwidth}
    \centering
    \includegraphics[width=1.08\linewidth]{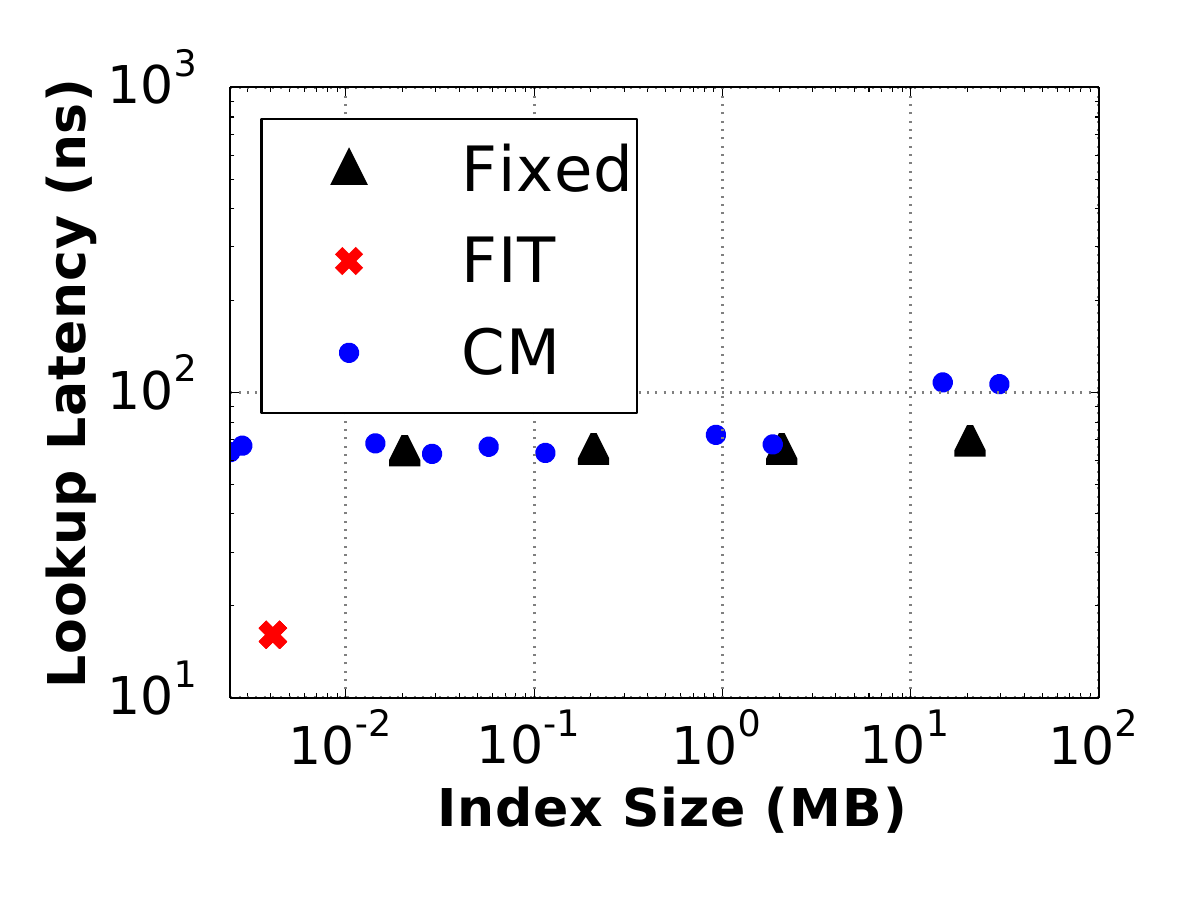}
    \subcaption{Linear Distribution}
    \label{fig:linear_distribution}
  \end{minipage}
  \hfill
  \begin{minipage}[b]{.52\columnwidth}
    \centering
    \includegraphics[width=0.97\linewidth]{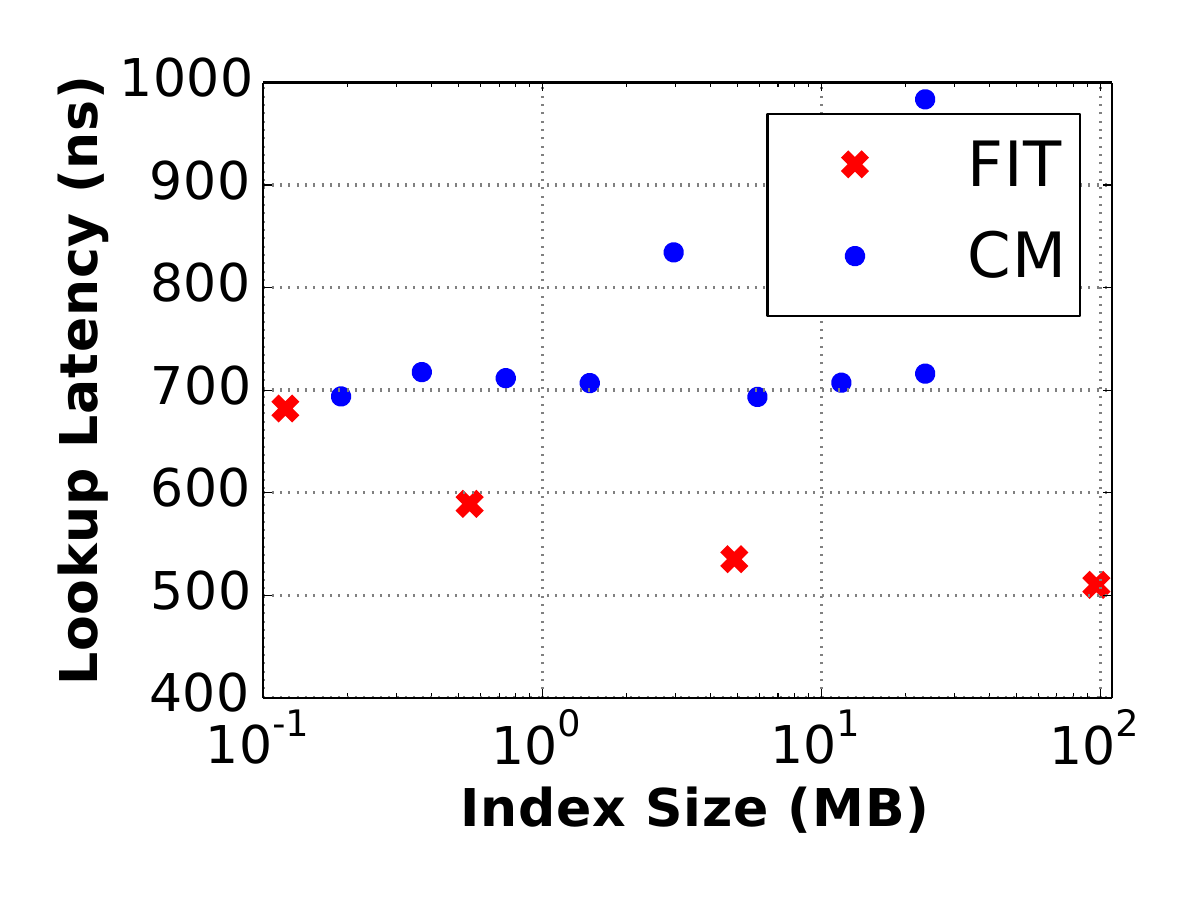}
    \subcaption{Weblogs}
    \label{fig:weblogs_cms}
  \end{minipage}
  \hfill
  \vspace{-5mm}
  \caption{Correlation Maps}
  \vspace{-4mm}
  \label{fig:cms}
\end{figure}

\subsection{Correlation Maps}
\label{sec:experiments:cms}
In the following, we compare \system{} to Correlation Maps (CMs) which are designed to exploit correlations that exist between attributes.
To ensure a fair comparison and to adapt the described techniques to build a clustered primary index, we assume that each tuple has an implicit position (e.g., ROWID). 
In addition to the design in the original paper \cite{correlation_maps} (e.g., bucketing both along the clustered and unclustered dimension), we implemented  additional optimizations not described in the paper since our use case has the additional knowledge that the unclustered attribute (i.e., timestamp) is sorted with respect to the clustered attribute (i.e., ROWID). 
For example, instead of storing several buckets for a given unclustered range (e.g., \{100-200\} $\rightarrow$ [b0,b1,b2,b3]) our implementation stores only the first and last bucket (e.g., \{100-200\} $\rightarrow$ [b0,b3]).
Additionally, when looking up a key, our implementation uses binary search within the entire region instead of searching each bucket individually.  We found that these two optimizations improved lookup performance/reduced the size of a CM.

First, to show that CMs are no more efficient for primary indexes than an index that uses fixed-size pages, we consider the simple case of indexing all integer values from 1 to 100M.
Although simple, this is not unusual since users often index monotonically increasing identifiers (e.g., customer ID).
Figure~\ref{fig:linear_distribution} shows the lookup latency for various index sizes (x-axis) for CMs, \system{}, and fixed-size paging.
As previously described, we vary the size of \system{} by selecting various error thresholds and use different page sizes/bucket sizes for CMs and B+ trees that use fixed-size paging.
As shown, CMs perform similar to B+ trees with fixed-size pages, since they use fixed-size buckets to partition the attribute domain.
\system{}, on the other hand, can use a single segment to represent this data and can locate to the exact position of any element using almost no space.

Next, we also used CMs also to build a primary key index on the Weblogs data set and compare it to our \system{}.
Figure~\ref{fig:weblogs_cms} shows index size (x-axis) vs. the lookup latency (y-axis) for both CMs and \system{} using 400M timestamps from the Weblogs dataset. 
Since \system{} creates variable-sized segments that better model the underlying data distribution (instead of the fixed-size bucketing approach that CMs use), \system{} is able to provide faster lookup performance using a smaller memory footprint.

\subsection{Range Queries}
\label{sec:experiments:range}
In addition to point queries, \system{} also supports range queries whereby an arbitrary number of tuples must be examined to compute the result.
Figure~\ref{fig:weblogs_range} shows the performance of a \system{} for range queries for both a sum and count aggregate for various selectivities using the Weblogs dataset.

Interestingly, to compute the result for a count query, a \system{} can subtract the start position from the end position of the range (i.e., a count aggregate over a range is essentially two point lookups), resulting in a constant lookup latency.
On the other hand, computing the sum of an attribute over a range requires examining every tuple in the range, resulting in significantly more work for larger ranges.

\begin{figure}
    \centering
    \includegraphics[width=.58\linewidth]{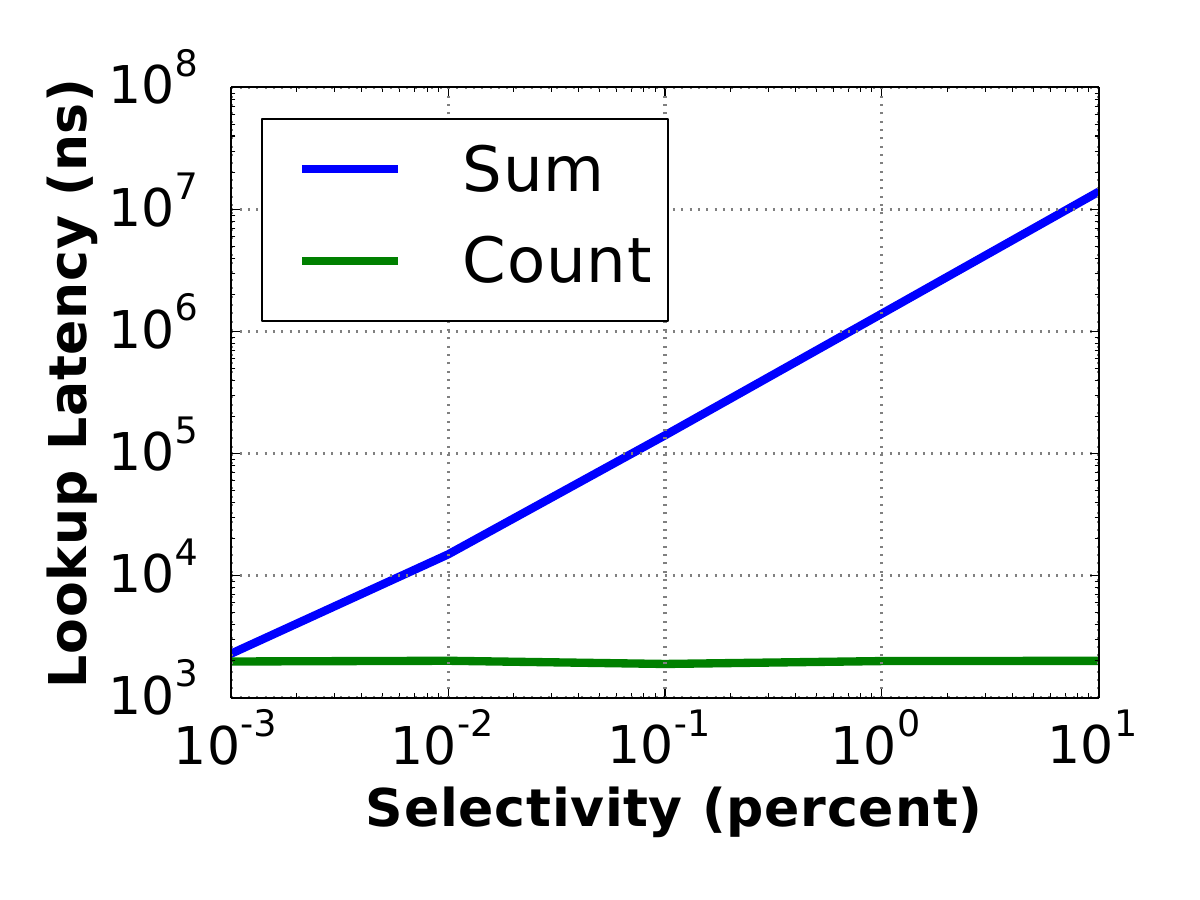}
    \caption{Range Queries}
  \vspace{-4mm}

    \label{fig:weblogs_range}
\end{figure}

\subsection{Varying Fill Factor}

As previously mentioned, the buffer size of a segment determines the amount of space that a segment reserves to hold new data items.
Once the segment's insert buffer reaches this threshold, the data from the segment and the segment's buffer are merged, and \system{} executes the previously described segmentation algorithm to generate new segments that satisfy the specified error threshold.

Therefore, in Figure~\ref{fig:fill_factor}, we vary the buffer size and measure the total throughput using the Weblogs dataset with an error threshold of $e=20,000$.
As shown, the size of the buffer can dramatically impact the write throughput of a \system{}.
More specifically, larger buffers result in fewer splitting operations, improving performance.
However, a buffer that is too large will result in longer lookup latencies (modeled in the cost model in Section~\ref{sec:cost_model}).

Therefore, the fill factor of an \system{} can be effectively used by a DBA to tune a \system{} to be more read or write optimized, depending on the workload.

\begin{figure}
\begin{center}
\includegraphics[width=0.65\columnwidth]{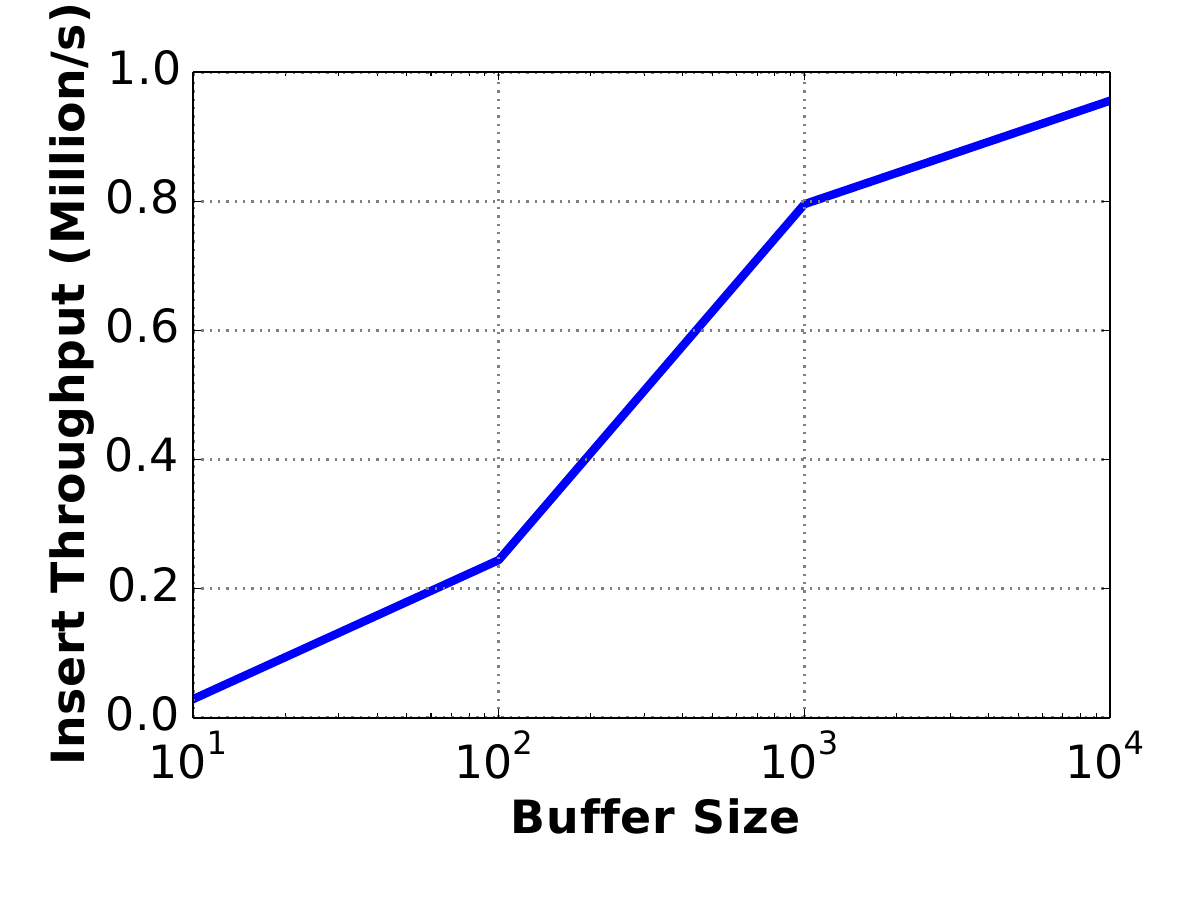}
\vspace{-2.5ex}
\caption{Insert Throughput / Varying Buffer Size}
\label{fig:fill_factor}
\end{center}
\end{figure}

\subsection{Lookup Breakdown}
As described in Section~\ref{sec:lookups}, a lookup involves two steps (i.e., locating the segment where a key belongs and then searching the segment's data in order to find the item within the segment).
Therefore, we examine the amount of time spent in each of these two steps for \system{} as well as an index that uses fixed-size paging for various error thresholds.

The results in Figure~\ref{fig:lookup_breakdown} show that in both cases the majority of time is spent searching the tree to find the page where the data item belongs for smaller error thresholds (and page sizes).
Since \system{} is able to leverage properties of the underlying data distribution in order to create variable sized segments, the resulting tree is significantly smaller.
Therefore, \system{} spends less time searching the tree to find the corresponding segment for a given key.

\begin{figure}
  \begin{minipage}[b]{.49\textwidth}
    \centering
    \includegraphics[width=0.8\linewidth]{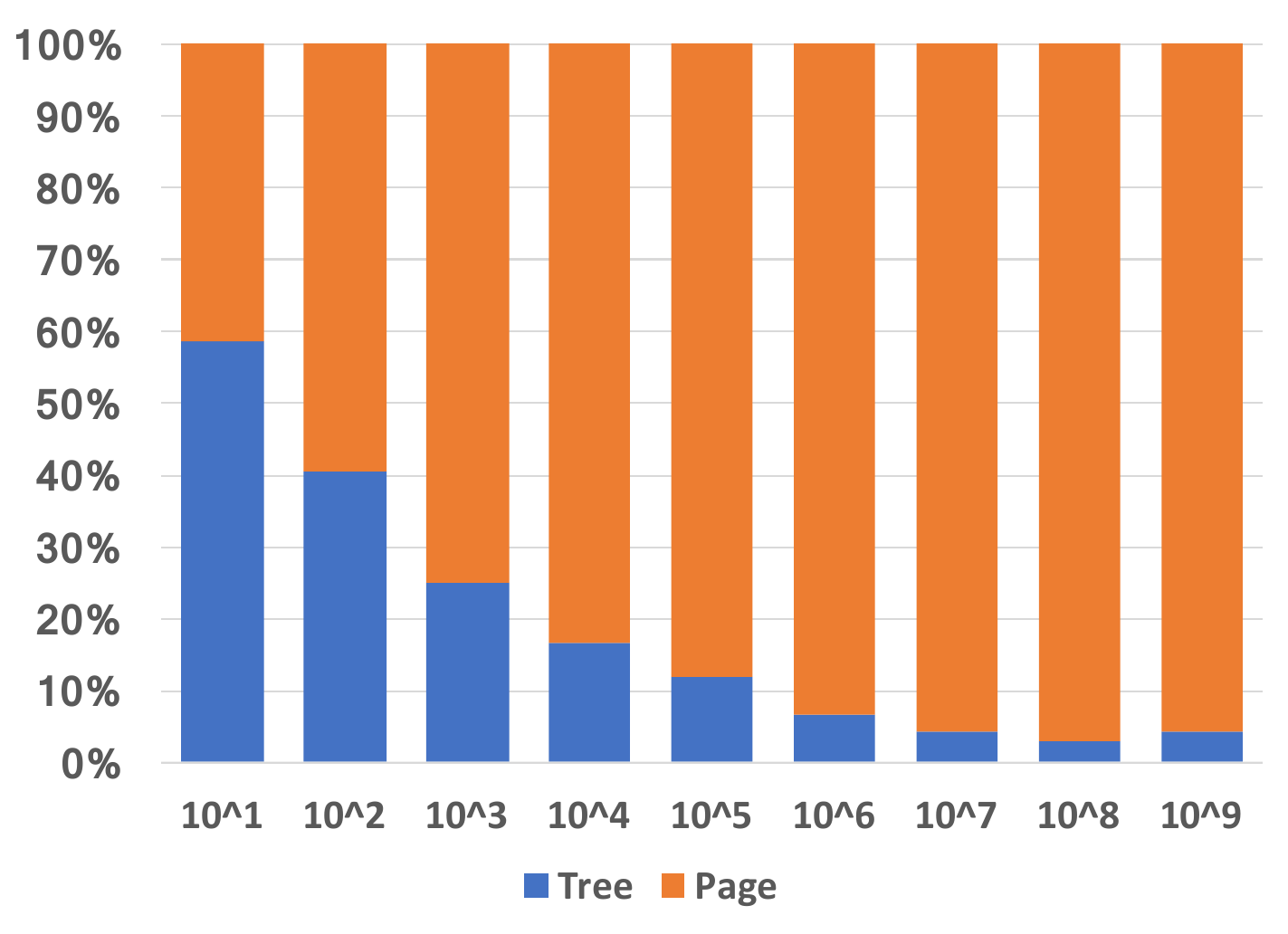}
    \subcaption{\system{}}
    \label{fig:atree_breakdown}
  \end{minipage}
  \hfill
  \begin{minipage}[b]{.49\textwidth}
    \centering
    \includegraphics[width=0.8\linewidth]{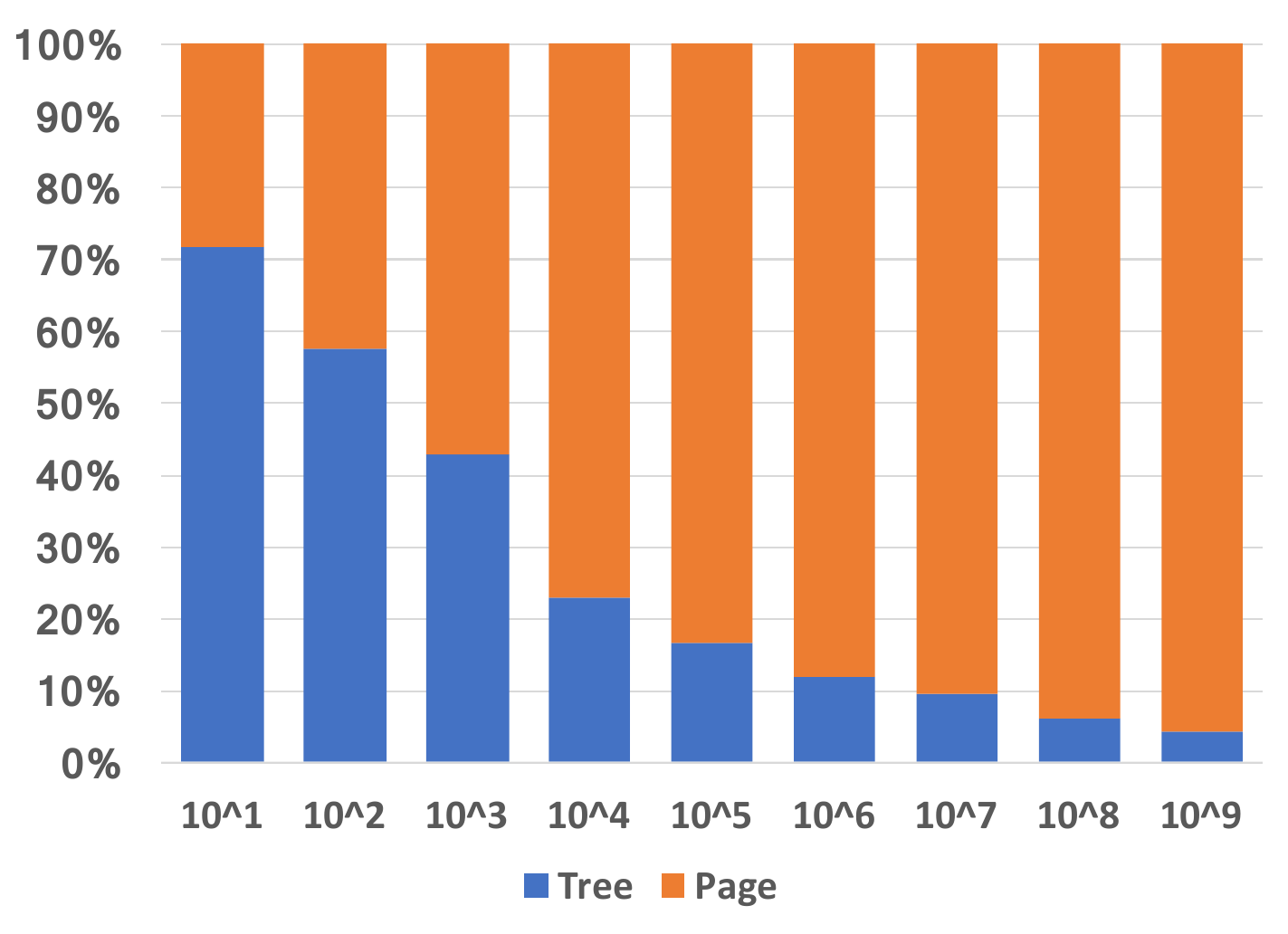}
    \subcaption{Fixed-size Index}
    \label{fig:fixed_breakdown}
  \end{minipage}
  \hfill
  \vspace{-1.5ex}
  \caption{Lookup Breakdown}
  \vspace{-2.5ex}
  \label{fig:lookup_breakdown}
\end{figure}
\end{appendices}

\end{document}